\pdfoutput=1

\documentclass[11pt]{article}

\title{Fully Dynamic Algorithms for Transitive Reduction}

\usepackage{authblk}
\author[1]{Gramoz Goranci}
\author[2]{Adam Karczmarz\thanks{Supported by the National Science Centre (NCN) grant no. 2022/47/D/ST6/02184. Work partially done at IDEAS NCBR, Poland.}}
\author[3]{Ali Momeni}
\author[4]{Nikos Parotsidis}
\affil[1]{Faculty of Computer Science, University of Vienna, Austria}
\affil[2]{University of Warsaw}
\affil[3]{Faculty of Computer Science, UniVie Doctoral School Computer Science DoCS, University of Vienna, Austria}
\affil[4]{Google Research, Switzerland}

\date{}

\usepackage{graphicx}
\usepackage{subcaption}

\graphicspath{{graphics/}{./graphics/}}

\usepackage[margin=1in]{geometry}

\usepackage{enumitem}
\setlist[itemize]{itemsep=.5pt, topsep=3pt}
\setlist[enumerate]{itemsep=.5pt, topsep=3pt}

\usepackage[
    dvipsnames   %
        ]{xcolor}

\definecolor{myOrange}{RGB}{230, 159, 0}
\definecolor{myLightBlue}{RGB}{86, 180, 233}
\definecolor{myGreen}{RGB}{0, 158, 115}
\definecolor{myYellow}{RGB}{240, 228, 66}
\definecolor{myDarkBlue}{RGB}{0, 114, 178}
\definecolor{myRed}{RGB}{213, 94, 0}
\definecolor{myPink}{RGB}{204, 121, 167}

\usepackage{amsmath}
\usepackage{amsthm}  %
\usepackage{amssymb}  %
\usepackage{mathtools}
\usepackage{enumitem}
\usepackage[]{bbold}

\usepackage{thmtools}
\usepackage{thm-restate}

\newtheorem{theorem}{Theorem}[section]
\newtheorem{corollary}[theorem]{Corollary}
\newtheorem{lemma}[theorem]{Lemma}
\newtheorem{fact}{Fact}[subsection]
\newtheorem{invariant}{Invariant}[section]

\usepackage[linesnumbered, boxruled, vlined]{algorithm2e}

\DontPrintSemicolon

\SetKwInput{Maintain}{Maintain}

\SetKwProg{Procedure}{Procedure}{}{}

\SetArgSty{textnormal}

\usepackage{xspace}

\newcommand{\counter}[1]{\textrm{c\(( #1 )\)}\xspace}
\newcommand{\touch}[1]{\textrm{t\(( #1 )\)}\xspace}
\newcommand{\graph}[1]{\text{\( G^{#1} \)\xspace}} %

\usepackage{mleftright}

\newcommand{\TR}[0]{\text{\( G^{t} \)}\xspace} 
 
\newcommand{\DAG}[1]{\textrm{{\normalfont DAG} \( #1 \)}\xspace}

\newcommand{\poly}[1]{\ensuremath{\operatorname{poly}\left( #1 \right)}\xspace}

\renewcommand{\textsc}[1]{\textnormal{\scshape #1}}

\newcommand{\nul}[0]{\textsf{null}\xspace}

\newcommand{\desc}[1]{\textrm{{\normalfont Desc}\( ^{ #1 }\)}\xspace}
\newcommand{\D}[1]{\textrm{{\normalfont D}\(^{ #1 } \)}\xspace}

\newcommand{\anc}[1]{\textrm{{\normalfont Anc}\( ^{ #1 }\)}\xspace}
\newcommand{\A}[1]{\textrm{{\normalfont A}\( ^{ #1 } \)}\xspace}

\newcommand{\s}[1]{\textrm{{\normalfont S}\( ^{ #1 } \)}\xspace}

\newcommand{\CC}[1]{\textrm{{\normalfont C}\( ^{ #1 } \)}\xspace}
\newcommand{\B}[1]{\textrm{{\normalfont B}\( ^{ #1 } \)}\xspace}

\newcommand{\E}[1]{\text{\( E^{#1} \)}\xspace} %

\newcommand{\T}[1]{\textrm{\normalfont \( T^{} \)}\xspace}

\newcommand{\Ei}[1]{\text{\( E_{#1} \)}\xspace} %
\newcommand{\Ed}{\text{\( E_{\mathrm{del}} \)}\xspace} %

\newcommand{\ts}[1]{\text{time\( ( #1 ) \)}\xspace} %

\newcommand{\F}[2]{\textrm{F\( [ #1, #2 ]  \)}\xspace}

\newcommand{\activee}[2]{\textrm{\normalfont active\( ^{  } [ #1 ] \)}\xspace} %
\newcommand{\inn}[2]{\textrm{\normalfont in\( ^{ } [ #1 ] \)}\xspace} 
\newcommand{\outt}[2]{\textrm{\normalfont{ out}\( _{#2} [ #1 ] \)}\xspace}

\newcommand{\scc}[2]{\textrm{{scc}\( _{#2} [ #1 ] \)}\xspace}

\newcommand{\updateC}[2]{\textsc{UpdateC\( ( #1
) \)}\xspace}

\newcommand{\updateP}[2]{\textsc{UpdateP\( ( #1
) \)}\xspace}

\newcommand{\edge}[2]{\ensuremath{ #1 #2 }\xspace}
\newcommand{\cc}[2]{\ensuremath{ c^{} (#1) }}

\newcommand{\parent}[2]{\ensuremath{ p^{} (#1) }}

\usepackage{tcolorbox}

\usepackage{listings}
\usepackage[
    backend=bibtex,    %
    maxcitenames=10,    %
    maxbibnames=10,    %
    maxalphanames=10,    
    maxnames=10,    %
    style=alphabetic,    %
    sorting=anyt,    %
    doi=false,    %
    isbn=false,    %
    eprint=false,    %
    url=false,    %
	backref=true,    %
	backrefstyle=none,    %
        ]{biblatex}

\newbibmacro{string+doi}[1]{%
    \iffieldundef{doi}{%
        \iffieldundef{url}{#1}{\href{\thefield{url}}{#1}}
        }{%
        \href{http://dx.doi.org/\thefield{doi}}{#1}
        }
    }

\DeclareFieldFormat*{title}{%
    \usebibmacro{string+doi}{\mkbibemph{#1}}
    }

\renewbibmacro{in:}{%
    \ifentrytype{article}{}{\printtext{\bibstring{in}\intitlepunct}}}

\newtoggle{bbx:datemissing}

\renewbibmacro*{date}{\toggletrue{bbx:datemissing}%
}

\renewbibmacro*{issue+date}{}%

\newbibmacro*{date:print}{%
    \setunit{\addcomma\addspace}    %
    \togglefalse{bbx:datemissing}%
    \printdate}

\renewbibmacro*{chapter+pages}{%
    \printfield{chapter}%
    \setunit{\bibpagespunct}%
    \printfield{pages}%
    \newunit
    \usebibmacro{date:print}%
    \newunit}

\renewbibmacro*{note+pages}{%
    \printfield{note}%
    \setunit{\bibpagespunct}%
    \printfield{pages}%
    \setunit{\addcomma\addspace}%
    \usebibmacro{date}%
    \newunit}

\renewbibmacro*{addendum+pubstate}{%
    \iftoggle{bbx:datemissing}
        {\usebibmacro{date:print}}
        {}%
    \printfield{addendum}%
    \newunit\newblock
    \printfield{pubstate}}

\appto{\bibsetup}{\sloppy}

\DefineBibliographyStrings{english}{%
    page={page},%
    pages={pages},%
    }

\DefineBibliographyStrings{english}{
  backrefpage = {cited on page},
  backrefpages = {cited on pages},
}

\usepackage{xurl}

\usepackage[
    breaklinks=true,    %
    colorlinks=true,    %
    linkcolor=myDarkBlue,    %
    urlcolor=myRed,    %
    citecolor=myGreen,    %
    bookmarks=true,    %
    bookmarksopenlevel=2,    %
    ]{hyperref}

\usepackage[nameinlink]{cleveref}

\crefname{problem}{problem}{problems}
\crefname{claim}{claim}{claims}
\crefname{fact}{fact}{facts}
\crefname{algorithm}{algorithm}{algorithms}
\crefname{proof}{proof}{proofs}
\crefname{observation}{observation}{observations}
\crefname{invariant}{invariant}{invariants}

\usepackage[english]{babel}
\usepackage{fix-cm}
\usepackage[utf8]{inputenc}
\usepackage[T1]{fontenc}
\usepackage{amsfonts}
\usepackage{slantsc}

\newcommand{\arxivVsConference}[2]{#1}

\begin{document}

\maketitle

\arxivVsConference{
\pagenumbering{roman}
}{}

\begin{abstract}
 Given a directed graph $G$, a transitive reduction $G^t$ of $G$ (first studied by Aho, Garey, Ullman [SICOMP `72]) is a minimal subgraph of $G$ that preserves the reachability relation between every two vertices in $G$.
 
 In this paper, we study the computational complexity of transitive reduction in the dynamic setting.
 We obtain the first fully dynamic algorithms for maintaining a transitive reduction of a general directed graph undergoing updates such as edge insertions or deletions.
 Our first algorithm achieves $O(m+n \log n)$ amortized update time, which is near-optimal for sparse directed graphs, and can even support extended update operations such as inserting a set of edges all incident to the same vertex, or deleting an arbitrary set of edges. Our second algorithm relies on fast matrix multiplication and achieves $O(m+ n^{1.585})$ \emph{worst-case} update time.   

\end{abstract}

\arxivVsConference{
\newpage
\pagenumbering{arabic}
}{}

\section{Introduction}
Graph sparsification is a technique that reduces the size of a graph by replacing it with a smaller graph while preserving a property of interest. The resulting graph, often called a \emph{sparsifier}, ensures that the property of interest holds if and only if it holds for the original graph.
Sparsifiers have numerous applications, such as reducing storage needs, saving bandwidth, and speeding up algorithms by using them as a preprocessing step.
Sparsification has been extensively studied for various basic problems in both undirected and directed graphs~\cite{Aho:1972aa, althofer1993sparse, benczur1996approximating, cheriyan2000approximating, spielman2011spectral, spielman2008graph}. In this paper, we focus on maintaining a sparsifier for the notion of transitive closure in dynamic directed graphs.

Computing the transitive closure of a directed graph (digraph) is one of the most basic problems in algorithmic graph theory. Given a graph~$G=(V,E)$ with $n$ vertices and $m$ edges, the problem asks to compute for every pair of vertices $s,t$ on whether $t$ is reachable from $s$ in $G$. The efficient computation of the transitive closure of a digraph has received much attention over the past decades. 
In dense graphs, due to the problem being equivalent to Boolean Matrix Multiplication, the best known efficient algorithm runs in $O(n^\omega)$ time, where $\omega<2.371552$ \cite{strassen1969gaussian, coppersmith1982asymptotic, Williams12, williams2024new}. %
In sparse graphs, transitive closure can be trivially computed in $O(nm)$ time\footnote{Interestingly, shaving a logarithmic factor is possible here~\cite{Chan08}.}.

In their seminal work, Aho, Garey, and Ullman~\cite{Aho:1972aa} introduced the notion of transitive reduction; a \emph{transitive reduction} of a digraph $G$ is a digraph \TR on $V$ with fewest possible edges that preserves the reachability information between every two vertices in $G$.
Transitive reduction can be thought of as a sparsifier for transitive closure. 

While the transitive reduction is known to be uniquely defined for a directed acyclic graph (DAGs), it may not be unique for general graphs due to the existence of strongly connected components (SCCs). For each SCC $S$ there may exist multiple smallest graphs on $S$ that preserve reachability among its vertices. One example of such a graph is a directed cycle on the vertices of $S$. %
Significantly, \cite{Aho:1972aa} showed that computing the transitive reduction of a directed graph requires asymptotically the same time as computing its transitive closure.

It is important to note that a transitive reduction with an asymptotically smaller size than the graph itself is not guaranteed to exist even if we allow introducing auxiliary vertices: indeed, any bipartite digraph $G$ with $n$ vertices on both sides equals its transitive closure and one needs at least $n^2$ bits to uniquely encode such a digraph.
This is in contrast to e.g., equivalence relations such as strong connectivity where sparsification all the way down to linear size is possible.

In a DAG $G$, the same unique transitive reduction $\TR$ could be equivalently defined as the (inclusion-wise) minimal subgraph of $G$ preserving the reachability relation~\cite{Aho:1972aa}. In some applications, having a sparsifier that remains a subgraph of the original graph might be desirable.
Unfortunately, in the presence of cycles, if we insist on \TR being a subgraph of a $G$, then computing such a subgraph \TR of minimum possible size is NP-hard\footnote{$G$ has a Hamiltonian cycle iff \TR is a cycle consisting of all vertices of $G$}.
However, if we redefine $\TR$ to be simply an \emph{inclusion-wise minimal subgraph} of $G$ preserving its reachability, computing it becomes tractable again\footnote{In fact, an inclusion-wise minimal transitive reduction can have at most $n$ more edges than the minimum-size transitive reduction, hence the former is a $2$-approximation of the latter in terms of size.}, as a minimal strongly connected subgraph of a strongly connected graph can be computed in near-linear time~\cite{GibbonsKRST91, han1995computing}.
Throughout this paper, for our convenience, we will adopt this minimal subgraph-based definition of a transitive reduction~$G^t$ of a general digraph.  At the same time, we stress that all our algorithms can also be applied to the original ``minimum-size'' definition~\cite{Aho:1972aa} of $\TR$ after an easy adaptation of the way reachability inside SCCs is handled.

The transitive reduction of a digraph finds applications across a multitude of domains such as the reconstruction of biological networks~\cite{bovsnavcki2012efficient, 9460463} and other types of networks (e.g.,~\cite{pinna2013reconstruction, klamt2010transwesd, aditya2013algorithmic}), in code analysis and debugging~\cite{xu2006regulated, netzer1993optimal}, network analysis and visualization for social networks~\cite{1517831,clough2015transitive}, signature verification~\cite{HOU2015144}, serving reachability queries in database systems~\cite{jin2012scarab}, large-scale parallel machine rescheduling problems~\cite{meng2023transitive}, and many more. 
In some of these applications~(e.g.,~\cite{bovsnavcki2012efficient,9460463}), identifying and eliminating edges \emph{redundant} from the point of view of reachability\footnote{i.e.\ that can be removed from the graph without affecting the transitive closure.} is more critical than reducing the size of the graph and, consequently, the space required to store it. 

In certain applications, one might need to compute the transitive reduction of dynamically evolving graphs, where nodes and edges are being inserted and deleted from the graph and the objective is to efficiently update the transitive reduction after every update. 
 The naive way to do that is to recompute it from scratch after every update, which has total update time $O\left(m \cdot \min(n^\omega,nm)\right)$, or in other words, the algorithm has $O\left(\min(n^\omega,nm)\right)$ amortized update time.
 This is computationally very expensive, though.
 It is interesting to ask whether a more efficient approach is possible.

The concept of dynamically maintaining the sparsifier of a graph is not new.
Sparsifiers for many graph properties have been studied in the dynamic setting, where the objective is to dynamically maintain a sparse certificate as edges or vertices are being inserted and/or deleted to/from the underlying dynamic graph. To the best of our knowledge, apart from transitive reduction, other studies have mainly focused on sparsifiers for dynamic \emph{undirected} graphs. %

In this paper, we study fully dynamic sparsifiers for reachability; that is, for one of the most basic notions in directed graphs. 
In particular, we continue the line of work initiated by La Poutr{\'e} and van Leeuwen~\cite{lapoutre1988} who were the first to study the maintenance of the transitive reduction in the partially dynamic setting, over 30 years ago. They presented an incremental algorithm with total update time $O(mn)$, and a decremental algorithm with total update time that is $O(m^2)$ for general graphs, and $O(mn)$ for DAGs.

\subsection{Our results}
We introduce the first fully dynamic data structures designed for maintaining the transitive reduction in digraphs. These data structures are tailored for both DAGs and general digraphs, and are categorized based on whether they offer worst-case or amortized guarantees on the update time.

\paragraph*{Amortized times for handling updates.} Our first contribution is two fully dynamic data structures for maintaining the transitive reduction of DAGs and general digraphs, each achieving roughly linear amortized update time on the number of edges. Both data structures are combinatorial and deterministic, with their exact guarantees summarized in the theorems below. 
\begin{restatable}{theorem}{dag}\label{thm: OurResultsDag1}
    Let $G$ be an initially empty graph that is guaranteed to remain acyclic throughout any sequence of edge insertions and deletions.  Then, there is a fully dynamic deterministic data structure that maintains the transitive reduction $G^t$ of $G$ undergoing a sequence of edge insertions centered around a vertex or arbitrary edge deletions in $O(m)$ amortized update time, where $m$ is the current number of edges in the graph.
\end{restatable}

For general sparse digraphs, we obtain a much more involved data structure, where we  pay an additional logarithmic factor in the update time.

\begin{restatable}{theorem}{general}\label{thm: OurResultsDag2}
    Given an initially empty general digraph $G$, there is a fully dynamic deterministic data structure that supports edge insertions centered around a vertex and arbitrary edge deletions, and maintains a transitive reduction $G^t$ of $G$ in $O(m + n \log n)$ amortized update time, where $m$ is the current number of edges in the graph.
\end{restatable}

In fact, the data structures from \Cref{thm: OurResultsDag1,thm: OurResultsDag2} support more general update operations: 1) insertions of a set of any number of edges, all incident to the same single vertex, and 2) the deletion of an arbitrary set of edges from the graph. Note that these \emph{extended} update operations are more powerful compared to the single edge insertions and deletions supported by more traditional dynamic data structures. For further details, we refer the reader to \Cref{subsec:combinatorial_dag} and~\Cref{subsec:combinatorial_general}. 

For sparse digraphs, our dynamic algorithms supporting insertions in $O(m + n \log n) = O(n \log n)$ are almost optimal, up to a $\log n$ factor. This is because a polynomially faster dynamic algorithm would lead to an improvement in the running time of the best-known static $O(n^2)$ algorithm for computing the transitive reduction of a sparse graph, which would constitute a major breakthrough.

Observe that within $O(m)$ amortized time spent on updating the data structure, one can explicitly list each edge of the maintained transitive reduction which is guaranteed to have at most $m$ edges. This is why \Cref{thm: OurResultsDag1,thm: OurResultsDag2} do not come with separate query bounds.

\paragraph*{Worst-case times for handling updates.}  

Our second contribution is another pair of fully dynamic data structures that maintain the transitive reduction of DAGs and general digraphs. These data structures achieve \emph{worst-case} update time (on the number of nodes) for vertex updates and sub-quadratic \emph{worst-case} update time for edge updates. This is as opposed to \Cref{thm: OurResultsDag1,thm: OurResultsDag2}, where the worst-case cost of a single update can be as much as $O(nm) = O(n^3)$. 

Both data structures rely on fast matrix multiplication and are randomized, with their exact guarantee summarized in the theorem below.

\begin{theorem} \label{thm: ourResultsAlgebraic1}
Let $G$ be a graph that is guaranteed to remain acyclic throughout any sequence of updates. Then, there are randomized Monte Carlo data structures for maintaining the transitive reduction $G^t$ of $G$
\begin{itemize}
    \item[(1)] in $O(n^2)$ worst-case update time for vertex updates, and
    \item[(2)] in $O(n^{1.528} + m)$ worst-case update time for edge updates.
\end{itemize}
Both data structures output a correct transitive reduction with high probability. They can be initialized in $O(n^\omega)$ time. %
\end{theorem}

For general digraphs, the runtime guarantees for vertex updates remain the same, whereas for edge updates, we incur a slightly slower sub-quadratic update time.

\begin{theorem} \label{thm: ourResultsAlgebraic2}
Given a general digraph $G$, there are randomized Monte Carlo data structures for maintaining the transitive reduction $G^t$ of $G$
\begin{itemize}
    \item[(1)] in $O(n^2)$ worst-case update time for vertex updates, and
    \item[(2)] in $O(n^{1.585} + m)$ worst-case update time for edge updates.
\end{itemize}
Both data structures output a correct transitive reduction with high probability. They can be initialized in $O(n^\omega)$ time. %
\end{theorem}

All our data structures require $O(n^2)$ space, similarly to the partially dynamic data structures proposed by~\cite{lapoutre1988}. Going beyond the quadratic barrier in space complexity is known to be a very challenging task in data structures for all-pairs reachability. For example, to the best of our knowledge, it is not even known whether there exists a \emph{static} data structure with $O(n^{2-\epsilon})$ space and answering arbitrary-pair reachability queries in $O(m^{1-\epsilon})$ time.
Finally, recall that in certain applications eliminating redundant edges in a time-efficient manner is vital, and quadratic space is not a bottleneck. %

\subsection{Related work}

Due to their wide set of applications, sparsification techniques have been studied for many problems in both undirected and directed graphs.
For undirected graphs, some notable examples include sparsification for $k$-connectivity \cite{han1995computing, nagamochi1992linear}, shortest paths~\cite{althofer1993sparse, peleg1989graph,baswana2007simple}, cut sparsifiers~\cite{benczur1996approximating, benczur2015randomized, fung2011general}, spectral sparsifiers~\cite{spielman2008graph, spielman2011spectral}, and many more.
For directed graphs, applications of sparsification have been studied for reachability \cite{Aho:1972aa}, strong connectivity \cite{han1995computing, vetta2001approximating}, shortest paths~\cite{King:1999aa,roditty2008roundtrip}, $k$-connectivity \cite{georgiadis2016sparse, laekhanukit2012rounding, cheriyan2000approximating}, cut sparsifiers~\cite{cen2020sparsification}, spectral sparsifiers for Eulerian digraphs~\cite{cohen2016faster, sachdeva2023better}, and many more.

There is also a large body of literature on maintaining graph sparsifiers on dynamic undirected graphs. Examples of this body of work includes dynamic sparsifiers for connectivity \cite{holm2001polylogarithmic}, shortest paths \cite{baswana2012fully, forster2019dynamic, bernstein2021adeamortization}, cut and spectral sparsifiers~\cite{ittai2016onfully, bernstein2022fully}, and $k$-edge-connectivity~\cite{aamand2023optimal}. 

\subsection{Organization}
In \Cref{sec:prelims}, we set up some notation. \Cref{sec:overview} provides a technical overview of our algorithms for both DAGs and general graphs.
\arxivVsConference{}{
We then present the simpler data structures for DAGs in \Cref{sec:reduction,subsec:combinatorial_dag}. Due to space constraints, the detailed description of our data structures for general graphs, together with some proofs, is deferred to the appendix.
}
\arxivVsConference{
We then present the combinatorial data structures for DAGs and general graphs in \Cref{sec:combinatorial}, 
followed by the algebraic data structures in \Cref{sec:algebraic}.
}{}

\section{Preliminaries}\label{sec:prelims}
In this section, we introduce some notation and review key results on transitive reduction in directed graphs, which will be useful throughout the paper. 

\paragraph*{Graph Theory}
Let \( G = (V, E) \) be a directed, unweighted, and loop-less graph where \( |V| = n \) and \( |E| = m \).
For each edge \( xy \), we call \( y \) an \textit{out-neighbor} of \( x \), and \( x \) an \textit{in-neighbor} of \( y \).
A \textit{path} is defined as a sequence of vertices \( P =  \left\langle v_1, v_2, \dots , v_k \right\rangle \) where \( v_i v_{i+1} \in E \) for each \( i \in [k-1] \).
We call \( v_1 \) and \( v_k \) as the \textit{first} vertex and the \textit{last} vertex of \( P \), respectively.
The \textit{length} of a path \( P \), \( |P| \), is the number of its edges.
For two (possibly empty) paths \( P = \left\langle u_1, u_2, \dots , u_k \right\rangle \) and \( Q = \left\langle v_1, v_2, \dots , v_l \right\rangle \), where \( u_k = v_1 \), the concatenation of \( P \) and \( Q \) is the path obtained by identifying the last vertex of \( P \) with the first vertex of \( Q \).
i.e.\ the path \( \left\langle u_1, u_2, \dots , u_k = v_1, v_2, \dots , v_l \right\rangle \).
A \textit{cycle} is a path whose first and last vertices are the same, i.e., \( v_1 = v_k \).
We say that \( G \) is a \textit{directed acyclic graph} (DAG) if \( G \) does not have any cycle.
We say there is a path \( u \rightarrow v \) from \( u \) to \( v \) (or \( u \) can reach \( v \)) if there exists a path \( P = \left\langle v_1, v_2, \dots , v_k \right\rangle \) with \( v_1 = u \) and \( v_k = v \).

A graph \( H \) is called a \textit{subgraph} of \( G \) if \( H \) can be obtained from \( G \) by deleting (possibly empty) subsets of vertices and edges of \( G \).
For a set \( U \subseteq E \), we define \( G \setminus U \) as the subgraph of \( G \) obtained by deleting edges in \( U \) from \( G \). 
We also define \( G \setminus \edge{u}{v} = G \setminus \{ \edge{u}{v} \} \).

\paragraph*{Transitive Reduction} \label{subsec:TR}

A transitive reduction of a graph \(G = (V, E)\) is a graph \( \TR = (V, \tilde{E})\) with the fewest possible edges that preserves the reachability information between every two vertices in \(G\).
i.e., for arbitrary vertices \( u,v \in V \), there is a \( u \rightarrow v \) path in \( G \) iff there is a \( u \rightarrow v \) path in \TR.
Note that \TR may not be unique and might not necessarily be a subgraph of \(G\).

For a \DAG{G}, \( G \) has a unique transitive reduction \TR \cite{Aho:1972aa}.
They presented an algorithm to compute \TR by identifying and eliminating the \textit{redundant} edges.
We say that edge \( xy \in E\) is redundant if there is a directed path \( x \rightarrow y \) in \( G \setminus xy \).
\begin{theorem}[\cite{Aho:1972aa}] \label{th:DAG_old_paper}
Every \DAG{G} has a unique transitive reduction \TR that is a subgraph of~\( G \) and is computed by eliminating all redundant edges of \( G \).
\end{theorem}

For a general graph \( G \), a transitive reduction \TR can be obtained by replacing every strongly connected component (SCC) in \( G \) with a cycle and removing the redundant inter-SCC edges one-by-one \cite{Aho:1972aa}.
But if we insist on \TR being a subgraph of \( G \), then finding \TR is NP-hard since \( G \) has a Hamiltonian cycle iff \TR is a cycle consisting of all vertices of \( G \).
To overcome this theoretical obstacle, we consider a \textit{minimal} transitive reduction \TR of \( G \).
Given a graph \( G \), we call \TR a minimal transitive reduction of \( G \) if removing any edge from the subgraph \TR results in \TR no longer being a transitive reduction of \( G \).
For the rest of this paper, we assume \TR is a subgraph of \( G \).  At the same time, we once again stress this is merely for convenience and all our algorithms can be easily adapted to maintain the minimum-size reachability preserver with SCCs replaced with cycles.

\section{Technical overview}\label{sec:overview}
Dynamic transitive closure has been extensively studied, with
efficient combinatorial~\cite{Roditty08, Roditty:2016aa} and algebraic~\cite{BrandNS19, Sankowski04} data structures known.
As mentioned before, computing  transitive reduction is closely
related to computing the transitive closure~\cite{Aho:1972aa}.
This is why, in this work, we adopt the general approach of reusing some of the technical ideas developed in the prior literature on dynamic transitive closure.
The main challenge lies in our goal to achieve near-linear update time in the number of edges $m$ and \emph{constant} query time,
or explicit maintenance of the transitive reduction.
Existing dynamic transitive closure data structures such as those in~\cite{Roditty08, Sankowski04}
have either $O(n^2)$ update time and $O(1)$ arbitrary-pair query time (which is optimal
if the reachability matrix is to be maintained explicitly),
or polynomial query time~\cite{BrandNS19, Roditty:2016aa, Sankowski04}.

To maintain the transitive reduction, we do not need to support arbitrary reachability queries. Instead, we focus on maintaining specific reachability information between $m$ pairs of vertices connected by an edge. This reachability information, however, is more sophisticated than in the case of transitive closure.

\subsection{DAGs.} Let us first consider the simpler case when $G$ is a DAG.
Recall (\Cref{th:DAG_old_paper}) that the problem boils down to identifying redundant edges. To test the redundancy of an edge $uv$,
we need to maintain information on whether a $u\to v$ path exists in the graph $G\setminus uv$.

\subsubsection{A reduction to the transitive closure problem}
In acyclic graphs, a reduction resembling that of~\cite{Aho:1972aa} (see~\Cref{l:dag-reduction}) carries quite well to the fully dynamic setting.
Roughly speaking, one can set up a graph $G'$ with two additional copies of every vertex and edge, so that for all $u,v\in V$, paths $u''\to v''$ between the respective second-copy vertices $u'',v''$ in~$G'$ correspond precisely to \emph{indirect} paths $u\to v$ in $G$, i.e.,
$u\to v$ paths avoiding the edge $uv$.
Clearly, an edge $xy$ is redundant if indirect $x\to y$ paths exist.
As a result, one can obtain a dynamic transitive reduction data structure by setting up a dynamic transitive closure data structure on $G$ and issuing $m$ reachability queries after each update.
The reduction immediately implies an $O(n^2)$ worst-case update bound (Monte Carlo randomized) and $O(n^2)$ amortized update bound (deterministic)
for the problem in the acyclic case via~\cite{Roditty08, Sankowski04}, even in the case of vertex updates. This is optimal for dense acyclic graphs, at least if one is interested in maintaining the reduction explicitly.\footnote{Indeed, consider a graph $G=(V_1\cup V_2\cup \{s,t\},E)$ with $E=(V_1\times V_2)\cup \{vs:v\in V_1\}\cup \{tv:v\in V_2\}$, where $n=|V_1|=|V_2|$. No edges in this acyclic graph are redundant. However, adding the edge $st$ makes all $n^2$ edges $V_1\times V_2$ redundant. Adding and removing $st$ back and forth causes $\Theta(n^2)$ amortized change in the transitive reduction.}

The reduction works best if the transitive closure data structure can maintain some (mostly fixed) $m$ reachability pairs $Y\subseteq V\times V$ of interest more efficiently than via issuing $m$ reachability queries. In our use case, the reachability pairs $Y$ correspond to the edges of the graph, so an update of $G$ can only change $Y$ by one pair in the case of single-edge updates, and by $n$ pairs sharing a single vertex in the case of vertex updates to $G$.
Some algebraic dynamic reachability data structures based on dynamic matrix inverse~\cite{Sankowski04, DBLP:conf/cocoon/Sankowski05} come with such a functionality out of the box. By applying these data structures on top of the reduction, one can obtain an $O(n^{1.528}+m)$ worst-case bound for updating the transitive reduction of an acyclic digraph (see~\Cref{thm:algebraic-dag} for details). Interestingly, the $n^{1.528}$ term (where the exponent depends on the current matrix multiplication exponent~\cite{williams2024new} and simplifies to $n^{1.5}$ if $\omega=2$) typically emerges in connection with single-source reachability problems, and $O(n^{1.528})$ is indeed conjectured optimal for fully dynamic single-source reachability~\cite{BrandNS19}. One could say that our algebraic transitive reduction data structure for DAGs meets a fairly natural barrier.

\subsubsection{A combinatorial data structure with amortized linear update time bound}
The algebraic reachability data structures provide worst-case guarantees but are more suitable for denser graphs.
In particular, they never achieve near-linear in $n$ update bounds, even when applied to sparse graphs.
On the other hand, some combinatorial fully dynamic reachability data structures, such as~\cite{RodittyZ08}, do achieve near-linear in $n$ update bound.
However, these update bound guarantees (1) are only amortized, and (2) come with a non-trivial polynomial query procedure that does not inherently support maintaining a set $Y$ of reachability pairs of interest.
This is why relying on the reduction (\Cref{l:dag-reduction}) does not immediately improve the update bound for sparse graphs. 
Instead, we design a data structure tailored specifically to maintain the transitive reduction of a DAG (\Cref{thm: OurResultsDag1}).
We later extend it to 
general graphs, ultimately obtaining \Cref{thm: OurResultsDag2}.

First of all, we prove that
given a source vertex $s$ of a \DAG{G}, one can efficiently maintain, subject to edge deletions, whether an indirect $s\to t$ path 
exists in~$G$ for every outgoing edge $st\in E$ (\Cref{lem:extended_italiano}). 
This ``indirect paths'' data structure 
is based on extending the decremental single-source
reachability data structure of~\cite{Italiano:1988aa}, which is a classical combinatorial data structure with an optimal $O(m)$ total update time.

Equipped with the above, we apply the strategy used in the reachability data structures of~\cite{RodittyZ08, Roditty:2016aa}: every time an insertion of edges centered at a vertex $z$ issued, we build a new decremental single-source indirect paths data structure $D_z$ ``rooted'' at $z$ and initialized with the current snapshot of $G$.
Such a data structure will not accept any further insertions, so the future insertions to $G$ are effectively ignored by~$D_z$.
Intuitively, the responsibility of~$D_z$ is to handle (i.e., test the indirectness of) paths in the \emph{current} graph $G$ \emph{whose most recently
updated vertex is $z$}.\footnote{In the transitive closure data structures~\cite{RodittyZ08, Roditty:2016aa},
the concrete goal of an analogous data structure rooted at $z$ is to efficiently query for the existence of a path $x\to y$, where \( z \) is the most recent updated vertex along the path.} 
This way, handling each path in $G$ is delegated to precisely one maintained
decremental indirect paths data structure rooted at a single vertex. 

Compared to the data structures of~\cite{RodittyZ08, Roditty:2016aa},
an additional layer of global bookkeeping is used to aggregate the counts of alternative
paths from the individual data structures $D_z$.
That is, for an arbitrary \( z\in V \), \( D_z \) contributes to the counter iff it contains an alternative path. Using~this~technique,
 if the count is zero for some edge $uv$, then $uv$ is not redundant. This allows constant-time redundancy queries, or, more generally, explicitly maintaining the set of edges in the transitive reduction.

We prove that all actions we perform can be charged to the $O(m)$ initialization time of the decremental data structures $D_z$. 
Since, upon each insert update (incident to the same vertex), we (re)initialize only one data structure $D_z$, the amortized update time of the data structure is $O(m)$. 

\subsection{General graphs}
For maintaining a transitive reduction of a general directed graph, the black-box reduction to fully dynamic transitive closure (\Cref{l:dag-reduction}) breaks. A natural idea to consider is to apply it to the acyclic \emph{condensation} of $G$ obtained by contracting the strongly connected components (SCCs).
However, a single edge update to $G$ might change the SCCs structure of $G$ rather dramatically.\footnote{For example, a single vertex of the condensation may break into~$n$ vertices forming a path (consider removing a single edge of an induced cycle).}
This, in turn, could enforce many single-edge updates to the condensation, not necessarily centered around a single vertex.
Using a dynamic transitive closure data structure (accepting edge updates) for maintaining the condensation in a black-box manner would then be prohibitive: all the known fully dynamic transitive closure data structures have rather costly update procedures, which are at least linear in $n$. 
Consequently, handling the (evolving) SCCs in a specialized non-black-box manner seems inevitable here.

Nevertheless, using the condensation of \( G \) may still be useful. First, since we are aiming for near-linear (in $m$) update time anyway (recall that $O(m^{1-\epsilon})$ update time is likely impossible), the  $O(n)$-edge transitive reductions of the individual SCCs can be recomputed from scratch. This is possible since a minimal strongly connected subgraph of a strongly connected graph can be computed in $O(m+n\log{n})$ time~\cite{GibbonsKRST91}.

The above allows us to focus on detecting the redundant \emph{inter-SCC} edges $xy$, that is, edges connecting two different SCCs $X,Y$ of $G$. The edge $xy$ can be redundant for two reasons. First, if there exists an indirect $X\to Y$ path in the condensation of $G$,
then $xy$ is redundant by the arguments we have used for DAGs.
Second, if there are other \emph{parallel} edges
$x_1y_1,\ldots,x_ky_k$ such that $x_i\in X$ and $y_i\in Y$.
In the latter case, if there is no indirect path \( X \rightarrow Y \), then clearly the transitive reduction
contains precisely one of $xy,x_1y_1,\ldots,x_ky_k$. In other words, all these edges but one (not necessarily $xy$) are redundant.

\subsubsection{Extending the combinatorial data structure.}
Extending the data structure for DAGs to support SCCs involves addressing some technical challenges.
The decremental indirect path data structures need to be generalized to support detecting indirect paths in the condensation.
The approach of~\cite{Italiano:1988aa} breaks for general graphs, though.
One solution here would be to adapt the near-linear decremental single-source reachability for general graphs~\cite{BernsteinPW19}.
However, that data structure is randomized, slower by a few log factors, and much more complex.
 Instead, as in~\cite{RodittyZ08,Roditty:2016aa}, we take advantage of the fact that we always maintain $n$ decremental indirect paths data structures on a nested family of snapshots of $G$. 
 This allows us to apply an algorithm from~\cite{Roditty:2016aa} to compute how the SCCs decompose as a result of deleting some edges in \emph{all} the snapshots at once, in $O(m+n\log{n})$ time. 
Since the SCCs in snapshots decompose due to deletions, the condensations are not completely decremental in the usual sense: some intra-SCC edges may turn into inter-SCC edges and thus are added to the condensation.
Similarly, the groups of parallel inter-SCC edges may split, and some edges that have previously been parallel may lose this status. 
Nevertheless, we show that all these problems can be efficiently addressed within the amortized bound of $O(m+n\log{n})$ which matches that of the fully dynamic reachability data structure of~\cite{Roditty:2016aa}. Similarly as in DAGs, one needs to check $O(1)$ counters and flags to test whether some edge of $G$ is redundant or not; this takes $O(1)$ time. The details can be found in \arxivVsConference{\Cref{subsec:combinatorial_general}}{Appendix~\ref{subsec:combinatorial_general}}.

\subsubsection{Inter-SCC edges in algebraic data structures for general graphs.} 
As indicated previously, operating on the condensation turns out to be very hard when using the dynamic matrix inverse machinery~\cite{BrandNS19, Sankowski04}.
Intuitively, this is because these data structures model edge updates as entry changes in the adjacency matrix, and this is all the dynamic matrix inverse can accept.
It seems that if we wanted an algebraic data structure to maintain the condensation, we would need to update it explicitly by issuing edge updates. Recall that there could be $\Theta(n)$ such edge updates to the condensation per single-edge update issued to $G$, even if $G$ is sparse.
This is prohibitive, especially since any updates to a dynamic matrix inverse data structure take time superlinear in $n$.

To deal with these problems, we refrain from maintaining the condensation.
Instead, we prove an algebraic characterization of the redundancy of a group of parallel inter-SCC edges \linebreak ${F=\{u_1v_1,\ldots,u_kv_k\}}$ connecting two distinct SCCs $R,T$ of $G$.
Specifically, we prove that if $\tilde{A}(G)$ is the \emph{symbolic adjacency matrix}~\cite{Sankowski04} (i.e., a variant of the adjacency matrix with each~$1$ corresponding to an edge $uv$ is replaced with an independent indeterminate $x_{u,v}$), then in order to test whether~$F$ is redundant it is enough to verify a certain polynomial identity (\Cref{t:tr-matrix}) on some $k+1$ elements of the inverse $(\tilde{A}(G))^{-1}$ whose elements are degree $\leq n$ multivariate polynomials.
Consequently, through all groups~$F$ of parallel inter-SCC edges in $G$, we obtain
that $O(n+m)$ elements of the inverse have to be inspected to deduce which inter-SCC edges are redundant.

By a standard argument involving the Schwartz-Zippel lemma, in the implementation, we do not actually need to operate on polynomials (which may contain an exponential number of terms of degree $\leq n$).
Instead, for high-probability correctness it is enough to test which of the aforementioned identities hold for some random variable substitution from a sufficiently large prime field $\mathbb{Z}/p\mathbb{Z}$ where $p=\Theta(\poly{n})$.
Since the inverse of a matrix can be maintained explicitly in $O(n^2)$ worst-case time subject to row or column updates, this yields a transitive reduction data structure with $O(n^2)$ worst-case bound per vertex update (see \Cref{thm:algebraic-general}, item~(1)).

Obtaining a better worst-case bound for sparser graphs in the single-edge update setting is more challenging.
Unfortunately, the elements of the inverse used for testing the redundancy of a group~$F$ of parallel intra-SCC edges do not quite correspond to the individual edges of~$F$; they also depend, to some extent, on the global structure of $G$. Specifically, the aforementioned identity associated with $F$ involves elements $(r,u_1),\ldots,(r,u_k)$, $(v_1,t),\ldots,(v_k,t)$, and $(r,t)$ of the inverse, where $r,t$ are arbitrarily chosen \emph{roots} of the SCCs $R$ and $T$, respectively.

Dynamic inverse data structures with subquadratic update time~\cite{BrandNS19, Sankowski04} (handling single-element matrix updates) generally do not allow accessing arbitrary query elements of the maintained inverse in constant time; this costs at least $\Theta(n^{0.5})$ time.
They can, however, provide constant-time access to a specified \emph{subset of interest} $Y\subseteq V\times V$ of its entries, at the additive cost $O(|Y|)$ in the update bound. Ideally, we would like $Y$ to contain all the $O(m+n)$ elements involved in identities to be checked. However, the mapping of vertices $V$ to the roots of their respective SCCs may quickly become invalid due to adversarial edge updates causing SCC splits and merges. Adding new entries to $Y$ is costly, though: e.g., inserting $n$ new elements takes $\Omega(n^{1.5})$ time.

We nevertheless deal with the outlined problem by applying the hitting set argument~\cite{UY91} along with a heavy-light SCC distinction.
For a parameter $\delta\in (0,1)$, we call an SCC  heavy if it has $\Theta(n^{\delta})$ vertices, and light otherwise.
We make sure that at all times our set of interest $Y$ in the dynamic inverse data structure contains
\begin{enumerate}[label=(\arabic*),topsep=3pt]
\item the edges~$E$,
\item a sample $(S\times V)\cup (V\times S)$ of rows and columns for a random subset $S$ (sampled once) of size $\Theta(n^{1-\delta}\log{n})$, and 
\item all $O(n^{1+\delta})$ pairs $(u,v)$ such that $u$ and $v$ are both in the same \emph{light} SCC.
\end{enumerate}
This guarantees that the set $Y$ allows for testing all the required identities in $O(n+m)$ time, has size $\tilde{O}(n^{1+\delta}+n^{2-\delta}+m)$ and evolves by $O(n^{1+\delta})$ elements per single-edge update, all aligned in $O(n)$ small square submatrices. For an optimized choice of $\delta$, the worst-case update time of the data structure is $O(n^{1.585}+m)$, i.e., slightly worse than we have achieved for DAGs.

It is an interesting problem whether the transitive reduction of a general directed graph can be maintained within the natural $O(n^{1.528}+m)$ worst-case bound per update. The details can be found in \arxivVsConference{\Cref{sec:app_algebraic_general}}{Appendix~\ref{sec:app_algebraic_general}}.

\section{Reductions}\label{sec:reduction}

In this section, we provide a reduction from fully dynamic transitive reduction to fully dynamic transitive closure on DAGs.

\begin{lemma}\label{l:dag-reduction}
Let $G=(V,E)$ be a fully dynamic digraph that is always acyclic. Let\linebreak
${Y=\{(x_i,y_i):i=1,\ldots,k\}\subseteq V\times V}$.
Suppose there exists a data structure maintaining whether paths $x_i\to y_i$
exist in $G$ for $i=1,\ldots,k$ subject to either:
\begin{itemize}
    \item fully dynamic single-edge updates to $G$ and single-element updates to $Y$,
    \item fully dynamic single-vertex updates to $G$ (i.e. changing any number of edges incident to a single vertex $v$) and single-vertex updates to $Y$ (i.e., for some $v\in V$, inserting/deleting from $Y$ any number of pairs $(x,y)$ satisfying
    $v\in \{x,y\}$).
\end{itemize}
Let $T(n,m,k)$ be the (amortized) update time of the assumed data structure.

Then, there exists a fully dynamic transitive reduction data structure supporting
the respective type of updates to $G$ with $O(T(n,m,m))$ amortized update time.
\end{lemma}
\begin{proof}
Let $V'=\{v':v\in V\}$ and $V''=\{v'':v\in V\}$ be two copies of the vertex set~$V$.
Let $E'=\{uv':uv\in E\}$ and $E''=\{u'v'':uv\in E\}$.
Consider the graph
\begin{equation*}
    G'=(V\cup V'\cup V'',E\cup E'\cup E'').
\end{equation*}
Note that a single-edge (resp., single-vertex) update to $G$ can be translated to three updates
of the respective type issued to $G'$.

Let us now show that an edge $xy\in E$ is redundant in $G$ if and only if
there is a path $x\to y''$ in $G'$.
Since $G$ is a DAG, if $xy$ is redundant, there exists an $x\to y$ path in $G$
with at least two edges, say $Q\cdot wz\cdot zy$,
where $Q=x\to w$ (possibly $w=x$). By the construction
of $G'$, $wz',z'y''\in E(G')$, and $Q\subseteq G\subseteq G'$.
As a result, $x$ can reach $y''$
in $G'$.
Now suppose that there is a path
from $x$ to $y''$ in $G'$.
By the construction of $G'$,
such a path is of the form
$Q\cdot wz'\cdot z'y''$ (where $Q\subseteq G$)
since $V''$ has incoming edges only
from $V'$, and $V'\cup V''$ has incoming edges only
from $V$.
The $x\to y$ path $Q\cdot wz\cdot zy\subseteq G$ has at least two edges and $x\neq z$ (by acyclicity),
so $xy$ is indeed a redundant edge.

It follows that the set of redundant edges of $G$ can be maintained using a dynamic transitive closure data structure set up on the graph $G'$ with the set $Y$ (of size $k=m$) equal
to $\{uv'':uv\in E\}$.
To handle a single-edge update to $G$,
we need to issue three single-edge updates
to the data structure maintaining $G'$,
and also issue a single-element update to the set $Y$. Analogously 
for vertex updates: an update centered
at $v$ causes an update to $Y$ such that
all inserted/removed elements have one of its coordinates equal to $v$ or $v''$.
\end{proof}

\section{Combinatorial Dynamic Transitive Reduction} \label{sec:combinatorial}

\arxivVsConference{
We start by explaining the data structure for DAGs in \Cref{subsec:combinatorial_dag},
and then present the general case in \Cref{subsec:combinatorial_general}.
}{}

{
\let\section\subsection
\let\subsection\subsubsection

\section{Combinatorial Dynamic Transitive Reduction on DAGs} \label{subsec:combinatorial_dag}

In this section, we explain a data structure for maintaining the transitive reduction of a \DAG{G}, as summarized in \Cref{thm: OurResultsDag1} below.
The data structure supports \textit{extended} update operations,
i.e., it allows the deletion of an \textit{arbitrary} set of edges \Ed or the insertion of some edges \Ei{u} incident to a vertex \( u \), known as the center of the insertion.

\arxivVsConference{}{
The data structure for general graphs is explained in \Cref{subsec:combinatorial_general}.
}

\dag*

The data structure uses a straightforward idea to maintain \TR: an edge \edge{x}{y} does not belong to~\TR if \( y \) has an in-neighbor \( z \neq x \) reachable from \( x \).
Thus, one can maintain \TR by maintaining the reachability information for each vertex \( u \in V \), but naively maintaining them results in \( O(mn) \) amortized update time for the data structure. 

To improve the amortized update time to \( O(m) \), we maintain the reachability information on a subgraph \( \graph{u} = (V, \E{u}) \) for every vertex \( u \in V \) defined as the snapshot of \( G \) taken after the last insertion \Ei{u} centered around \( u \) (if such an insertion occurred).
The reachability information are defined as follows: \desc{u} is the set of vertices reachable from \( u \) in \graph{u} and \anc{u} is the set of vertices that can reach \( u \) in \graph{u}.

Note that subsequent edge insertions centered around vertices \emph{different} from $u$ do not change $\graph{u}$. 
However, edges that are subsequently deleted from \( G \) are also deleted from $\graph{u}$, ensuring that at each step, \( \E{u} \subseteq E \).
Therefore, \graph{u} undergoes only edge deletions while we decrementally maintain \desc{u} and \anc{u} until an insertion \Ei{u} centered around \( u \) happens and we reinitialize \graph{u}. 

To decrementally maintain \desc{u} and \anc{u}, our data structure uses an extended version of the decremental data structure of
\citeauthor{Italiano:1988aa}~\cite{Italiano:1988aa}, summarized in the lemma below.
See~\Cref{app:dag} for a detailed discussion.

\begin{restatable}{lemma}{italiano}
\label{lem:extended_italiano}
Given a \DAG{G = (V, E)} initialized with the set \desc{r} of vertices reachable from a root vertex \( r \) and the set \anc{r} of vertices that can reach \( r \),
there is a decremental data structure undergoing arbitrary edge deletions at each update that maintains \desc{r} and \anc{r} in \( O(1) \) amortized update time.

Additionally, the data structure maintains the sets \D{r} and \A{r} of vertices removed from \desc{r} and \anc{r}, resp., due to the last update and supports the following operations in \( O(1) \) time:
\begin{itemize}

\item 
\textsc{In\( (y) \):} Return \texttt{True} if \( y \neq r \) and \( y \) has an in-neighbor from \( \desc{r} \setminus r \), and \texttt{False} otherwise. 

\item
\textsc{Out\( (x) \):} Return \texttt{True} if \( x \neq r \) and \( x \) has an out-neighbor to \( \anc{r} \setminus r \), and \texttt{False} otherwise. 
\end{itemize}
\end{restatable}

The lemma below shows how maintaining \desc{u} and \anc{u} will be useful in maintaining \TR.
Recall that an edge \( \edge{x}{y} \in E \) is redundant if there exists a \( x \to y \) path in \( G \setminus xy \).

\begin{lemma} \label{lem:redundant}
Edge \edge{x}{y} is redundant in \( G \) iff one of the following holds.
\begin{enumerate}
\item[\normalfont{(1)}] There is a vertex \( z \notin \{ x, y \} \) such that \( x \in \anc{z} \) and \( y \in \desc{z} \) in \graph{z}. 
\item[\normalfont{(2)}] Vertex \( y \) has an in-neighbor \( z \in \desc{x} \setminus x  \) in \graph{x}. 
\item[\normalfont{(3)}] Vertex \( x \) has an out-neighbor \( z \in \anc{y} \setminus y \) in \graph{y}. 
\end{enumerate}
\end{lemma}
\begin{proof}%
We first show the ``if'' direction.
Suppose that \( \edge{x}{y} \) is redundant in \( G \).
Since $G$ is a DAG, and by the definition of redundant edges, there exists a directed path $P$ from $x$ to $y$ in $G \setminus {\edge{x}{y}}$ whose length is at least two. 
Let $c$ be a vertex on $P$ that has been a center of an insertion most recently. 
This insert operation constructed sets \desc{c} and \anc{c}. 
At the time of this insertion, all the edges of the path $P$ were already present in the graph, and thus $P$ exists in the graph \graph{c}.
Now, if $c \notin \{ x,y \} $, then $P$ is a concatenation of subpaths $x \to c$ and $c \to y$ in $\graph{c}$, or equivalently, Item~(1) holds. If $c=x$, then Item~(2) holds, and if $c = y$, then Item~(3) holds.

For the ``only-if'' direction, we only prove Item (2); items (1) and (3) can be shown similarly. 
To prove Item (2), let us assume that \( y \) has an in-neighbor \( z \in \desc{x} \setminus x \) in \graph{x}.
Our goal is to show that \( \edge{x}{y} \) is redundant in \( G \). 
To this end, since $z \notin \{ x,y \} $ and \graph{x} is a DAG, there exists a path $P$ from $x$ to $y$ in $\graph{x}$ that is a concatenation of the subpath $x \rightarrow z$ and the edge \edge{z}{y}. 
Note that $P$ has length at least two and does not contain the edge \edge{x}{y}. 
Since \( \E{x} \subseteq E \), \( P \) is also a path from \( x \) to \( y \) in \( G \setminus \edge{x}{y} \), which in turn implies that \( \edge{x}{y} \) is redundant in \( G \). 
\end{proof}

To incorporate~\Cref{lem:redundant} in our data structure, we define \counter{\edge{x}{y}} and \touch{\edge{x}{y}} for every edge \( \edge{x}{y} \in E \) as follows:
\begin{itemize}
\item 
The counter \( \counter{\edge{x}{y}} \) stores the number of vertices \( z \notin \{ x, y \} \) such that \( \edge{x}{y} \in \E{z} \), \( x \in \anc{z} \), and \( y \in \desc{z} \) in \graph{z}.

\item
The binary value \touch{\edge{x}{y}} is set to \( 1 \) if either \( y \) has an in-neighbor \( z \in \desc{x} \setminus x \) in \graph{x} or \( x \) has an out-neighbor \( z \in \anc{y} \setminus y  \) in \graph{y}, and is set to \( 0 \) otherwise.
\end{itemize}

Note that for a redundant edge $\edge{x}{y}$ in $G$, a vertex $z$ contributes towards \counter{\edge{x}{y}} iff $\edge{x}{y} \in \E{z}$. 
If no such $z$ exists, then \edge{x}{y} has become redundant due to the last insertion being centered around \( x \) or \( y \), which in turn implies \( \touch{\edge{x}{y}} = 1 \). 
Combining these two facts, we conclude the following invariant.
\begin{invariant} \label{inv:dag}
An edge \( \edge{x}{y} \in E \) belongs to the transitive reduction \TR iff \( \counter{\edge{x}{y}} = 0 \) and \( \touch{\edge{x}{y}} = 0 \).
\end{invariant}
In the rest of the section, we show how to efficiently maintain \counter{\edge{x}{y}} and \touch{\edge{x}{y}} for every \( \edge{x}{y} \in E \).

\subsection{Edge insertions} \label{subsub:insertion_dag}
After the insertion of \Ei{u} centered around $u$, the data structure updates \graph{u}, while leaving every other graph \graph{v} is \emph{unchanged}.

We simply use a graph search algorithm to recompute \desc{u} and \anc{u}, and the sets \CC{u} and  \B{u} of the new vertices added to \desc{u} and \anc{u}, respectively, due to the insertion of \Ei{u}.

To maintain \counter{\edge{x}{y}} for an edge \( \edge{x}{y} \in E \), we only need to examine the contribution of \( u \) in \counter{\edge{x}{y}} as only \graph{u} has changed. 
By definition, every edge \edge{x}{y} touching \( u \), i.e., with \( x = u \) or \( y = u \), does not contribute in \counter{\edge{x}{y}}.
Note that \E{u} consists of the edges before the update and the newly added edges,  which leads to distinguishing the following cases.

\begin{enumerate}

\item[(1)]
Suppose that \edge{x}{y} is not touching \( u \) and has already existed in \E{u} before the update.
Then, \counter{\edge{x}{y}} will increase by one only if the update makes \( x \) to reach \( y \) through \( u \). 
i.e., there is no path \( x \to u \to y \) before the update  (\( x \notin \anc{u} \setminus \B{u} \) or \( y \notin \desc{u}\setminus \CC{u} \)), but there is at least one afterwards (i.e., \( x \in \anc{u} \) and \( y \in \desc{u} \)).

\item[(2)]
Suppose that \edge{x}{y} is not touching \( u \) and has added to \E{u} after the update.
Since this is the first time we examine the contribution of \( u \) towards \counter{\edge{x}{y}}, we increment \counter{\edge{x}{y}} by one if \( x \) can reach \( y \) through \( u \) (i.e., \( x \in \anc{u} \) and \( y \in \desc{u} \)).
\end{enumerate}

We now maintain \touch{\edge{x}{y}} for an edge \( \edge{x}{y} \in E \). 
By definition, \( \touch{\edge{x}{y}} \) could be affected only if \edge{x}{y} touches \( u \).
Since \graph{u} undergoes edge insertions, \touch{\edge{x}{y}} can only change from $0$ to~$1$.
For every edge \( \edge{x}{y} \in E \) touching \( u \), we set $\touch{\edge{x}{y}} \gets 1$ if one of the following happen:
\begin{itemize}
\item[(a)] 
$x=u$ and \( y \) has an in-neighbor \( z \neq u \) in \graph{u} reachable from \( u \) (i.e., if $\textsc{In}(y)$ reports $\texttt{True}$), or

\item[(b)]
$y=u$ and \( x \) has an out-neighbor \( z \neq u \)  in \graph{u} that can reach \( y \) (i.e., if $\textsc{Out}(y)$ reports $\texttt{True}$).
\end{itemize}
Note that both cases above try to insure the existence of a path \( x \to z \to y \) when \( x = u \) or \( y = u \).

\begin{lemma} \label{lem:insertion}
After each insertion, the data structure maintains the transitive reduction \TR of \( G \) in \( O(m) \) worst-case time, where \( m \) is the number of edges in the current graph \( G \).
\end{lemma}
\begin{proof}%
Suppose that the insertion of edges \Ei{u} has happened centered around \( u \).

\underline{Correctness:}
by construction, all possible scenarios for the edge \( xy \) are covered and the values \counter{\edge{x}{y}} and \touch{\edge{x}{y}} are correctly maintained after each insert update with respect to \graph{u} for every edge \( \edge{x}{y} \in E \).
Since \graph{u} is the only graph changing during the   update, it follows that \counter{\edge{x}{y}} and \touch{\edge{x}{y}} are correctly maintained with respect to \graph{v}, \( v \in V \).
The correctness immediately follows from \Cref{inv:dag} and \Cref{lem:redundant}.

\underline{Update time:}
the update time is dominated by (i) the time required to initialize the data structure of \Cref{lem:extended_italiano} for \graph{u}, and (ii) the time required to maintain \counter{\edge{x}{y}} and \touch{\edge{x}{y}} for every edge \( \edge{x}{y} \in E \). 
By \Cref{lem:extended_italiano}, the time for (i) is at most $O(m)$. 
To bound the time for (ii), note that $\counter{\edge{x}{y}}$ for any edge \( \edge{x}{y} \in E \) can be updated in $O(1)$ time as we only inspect the contribution of \( u \) towards $\counter{\edge{x}{y}}$ as discussed before.
\Cref{lem:extended_italiano} guarantees that each $\textsc{In}(\cdot)$ or $\textsc{Out}(\cdot)$ is supported in time $O(1)$, which in turn implies that $\touch{\edge{x}{y}}$ can be updated in $O(1)$ time.
As there are at most $m$ edges in $G$, maintaining \counter{\cdot} and \touch{\cdot} takes at most $O(m)$.
\end{proof}

\subsection{Edge deletions}  
\label{subsub:deletion_dag}

After the deletion of \Ed, the data structure passes the deletion to \emph{every} \graph{u}, \( u \in V \).
Let \D{u} and \A{u} denote the sets of vertices removed from \desc{u} and \anc{u}, resp., due to the deletion of \Ed.

We decrementally maintain \desc{u} and \anc{u}, and the sets \D{u} and \A{u} using the data structure of \Cref{lem:extended_italiano}, which is initialized last time \graph{u} was rebuilt due to an insertion centered around \( u \).

To maintain \counter{\edge{x}{y}} for any edge \( \edge{x}{y} \in E \), we need to cancel out every vertex \( z \notin \{ x, y \} \) that contained a path \( x \to z \to y \) in \graph{z} before the update but no longer has one.
i.e., \( x \in \anc{z} \cup \A{z} \) and \( y \in \desc{z} \cup \D{z} \) in \graph{z} and either \( x \in \A{z} \) or \( y \in \D{z} \).
This suggests that it suffices to subtract \counter{\edge{x}{y}} by one if \( x \) and \( y \) fall into one of the following disjoint cases.
\begin{enumerate}
\item 
\( x \in \A{z} \) and \( y \in \desc{z} \), or

\item
\( x \in \A{z} \) and \( y \in \D{z} \), or

\item
\( x \in \anc{z} \) and \( y \in \D{z} \).
\end{enumerate}

For cases (1) and (2) where \( x \in \A{z} \), we can afford to inspect every outgoing edge \( \edge{x}{y} \in \E{z} \) of \( x \) with \( y \neq z \), and subtract \counter{\edge{x}{y}} by one if \( y \in \desc{z} \cup \D{z} \).
For case (3) where \( y \in \D{z} \), we inspect every incoming edge \( \edge{x}{y} \in \E{z} \) of \( y \) with \( x \neq z \), and subtract \counter{\edge{x}{y}} by one if \( x \in  \anc{z} \).
 
To maintain \touch{\edge{x}{y}} for an edge \( \edge{x}{y} \in E \), we only need to inspect the updated graphs \graph{x} and \graph{y}.
Note that since the graphs are decremental, \touch{\edge{x}{y}} can only change from $1$ to $0$.
we check whether there is no path \( x \to z \to y \) in \graph{x} and \graph{y} passing through a vertex \( z \notin \{ x,y \} \)  by utilizing the data structure of \Cref{lem:extended_italiano}: if both $\textsc{In}(y)$ and $\textsc{Out}(x)$ return \texttt{False}, we set \( \touch{\edge{x}{y}} \gets 0 \).

\begin{lemma} \label{lem:deletion}
After each deletion, the data structure maintains the transitive reduction \TR of \( G \) in \( O(m) \) amortized time, where \( m \) is the number of edges in the current graph \( G \).
\end{lemma}
\begin{proof}%
Suppose that a deletion of edges \Ed has happened.

\underline{Correctness:} 
similar to handling edge insertions.
The correctness follows from \Cref{lem:redundant} and the observation that the  values \counter{xy} and \touch{xy} for every edge \( xy \in E \) are maintained correctly.

\underline{Update time:}
the update time is dominated by (i) the time required to decrementally maintain the data structures of \Cref{lem:extended_italiano} for \graph{v}, \( v \in V \), and (ii) the time required to update the values \counter{xy} and \touch{xy} for every edge \( xy \in E \).

By \Cref{lem:extended_italiano}, it follows that the total time required to maintain (i) over a sequence of \( m \) edge deletions is bounded by $O(m)$ for each graph.
Therefore, it takes \( O(n) \) amortized time to decrementally maintain all graphs.

To bound (ii), we first examine the total time required to maintain \counter{\edge{x}{y}} over a sequence of \( m \) edge deletions in a single graph \graph{z}.
Note that in \graph{z}, every vertex \( u \in V \) is inspected at most twice since it is deleted at most once from each set \desc{z} or \anc{z}.
During the inspection of \( u \) in \graph{z}, the data structure examines each incoming or outgoing edge \( e \) of \( u \) in \( O(1) \) time, and updates \counter{e} if necessary.
As explained before, the value $\counter{e}$ for a single edge \( e \) can be updated in $O(1)$ time. 
Thus, in \graph{z}, the time required for maintaining \counter{xy} for every \( xy \in E^z \) is bounded by
\(
O\left(  \sum _{u \in \D{z}} \deg(u) + \sum _{u \in \A{z}} \deg(u) \right),
\)
where \( \D{z} \) and \( \A{z} \) are the vertices that no longer belong to \desc{z} and \anc{z}, respectively, due to the deletion of \Ed.
Note that, during a sequence of \(m\) edge deletions, sets \D{z} and \A{z} partition the vertices of \graph{z}, and so the total time for maintaining \counter{xy}, \( xy \in E^z \), in \graph{z} is bounded by
\(
O\left(  \sum _{u \in V} \deg(u) + \sum _{u \in V} \deg(u) \right) = O\left( m \right).
\)
We conclude that maintaining \counter{xy}, \( xy \in E \), in all graphs \graph{\cdot} takes \( O(n) \) amortized time.

\Cref{lem:extended_italiano} guarantees that each in-neighbor or out-neighbor query can also be supported in $O(1)$ time, which in turn implies that \touch{xy} for a an edge \( xy \in E \) can be updated in $O(1)$ time.
As there are at most $m$ edges in the current graph $G$, updating \touch{xy} for all \( xy \in E \) costs $O(m)$ after each delete operation.
Since the number of delete operations is bounded by the number of edges that appeared in \( G \), we conclude that the total cost to maintain \touch{xy} for all \( xy \in E \) over all edge deletions is \( O(m^2) \), which bounds (ii).

Combining the bounds we obtained for (i) and (ii), we conclude that edge deletions can be supported in $O(m)$ amortized update time.
\end{proof}

\subsection{Space complexity}
 \label{subsubsec:space}

It remains to discuss the space complexity of the data structure.
Note that explicitly storing all graphs would require \( \Omega(n^2 + nm) \) space.
In the rest of this section, we sketch how to decrease the space to \( O(n^2) \) using a similar approach in~\cite{Roditty:2016aa}.

For every edge \( \edge{x}{y} \in E \), we define a \textit{timestamp} \ts{\edge{x}{y}}, attached to \edge{x}{y}, denoting its time of insertion into \( G \).
We maintain a single explicit adjacency list representation of $G$, so that the outgoing incident edges $E[v]$ of each $v$ are stored in increasing order by their timestamps. 
This adjacency list is shared by all the snapshots, and is easily maintained when edges of $G$ are inserted or removed: insertions can only happen at the end of these lists, and deletions can be handled using pointers into this adjacency list.

Let $\ts{v}$ denote the last time an insertion centered at $v$ happened.
Let us order the vertices $V=\{v_1,\ldots,v_n\}$ so that
$\ts{v_i}<\ts{v_j}$ for all $i<j$, i.e., from the least to most recently updated.
By the definition of snapshots,
we have
\[ \graph{v_1}\subseteq \graph{v_2}\subseteq \ldots \subseteq \graph{v_n}=G.\]

Note that for each $i$, the edges of $\graph{v_i}$ that are not in $\graph{v_{i-1}}$ are all incident to $v_i$ and have timestamps larger than timestamps of the edges in $\graph{v_{i-1}}$.
As a result, one could obtain the adjacency list representation of $\graph{v_i}$ by taking the respective adjacency list of $G$ and truncating all the individual lists $E[v]$ at the last edge $e$ with
$\ts{e}\leq \ts{v_i}$.
This idea gives rise to a \emph{virtual} adjacency list of $\graph{v_i}$, which requires only storing $\ts{v_i}$ to be accessed.
One can thus process $\graph{v_i}$ by using the global adjacency list for $G$ and ensuring to never iterate through the ``suffix'' edges in $E[v]$ that have timestamps larger than $\ts{v_i}$.
Using the timestamps, it is also easy to notify the required individual snapshots when an edge deletion in $G$ affects them.

Since all the auxiliary data structures for individual snapshots apart from their adjacency lists used $O(n)$ space, this optimization decreases space to $O(n^2)$.

\section{Combinatorial Dynamic Transitive Reduction on General Graphs} \label{subsec:combinatorial_general}

In this section, we explain our data structure when \( G \) is a general graph, as summarized in \Cref{thm: OurResultsDag2} below.

\general*

Analogous to \Cref{subsec:combinatorial_dag}, we maintain the reachability information on a decremental subgraph \( \graph{u} = (V, \E{u}) \) for every vertex \( u \in V \) defined as the snapshot of \( G \) taken after the last insertion \Ei{u} centered around \( u \) (if such an insertion occurred).
The reachability information are defined as follows: \desc{u} is the set of vertices reachable from \( u \) in \graph{u} and \anc{u} is the set of vertices that can reach \( u \) in \graph{u}.

The main difference here from DAGs is the existence of strongly connected components (SCCs) with arbitrary size.
We define the \textit{condensation} of \( G \) to be the DAG obtained by contracting each SCC of \( G \).
 Our goal is to maintain the \textit{inter-SCC} and \textit{intra-SCC} edges belonging to the transitive reduction \TR of \( G \).

To maintain the inter-SCC edges of \TR, we adapt our data structure on DAGs utilized on the condensation of \( G \),
equipped to address two challenges that did not arise during the designing of the data structure of \Cref{subsec:combinatorial_dag}: (i) each update could cause vertex splits or vertex merges in the condensation of \( G \), and (ii) the condensation of \( G \) consists of parallel edges.
To address this challenges, we use an extended version of the decremental data structure of \cite{Roditty:2016aa} explained in \Cref{app:general} and summarized in the lemma below.
This data structure maintains SCCs of \graph{u} for each \( u \in V \) and the set of parallel edges between the SCCs in \( G \), which will be used to address (i) and (ii), respectively.

\begin{restatable}{lemma}{roditty}
\label{lem:extended_roditty}

Given a directed graph \( G = (V, E) \) and the subgraphs \graph{r} for every \( r \in V \) initialized with the set \desc{r} of vertices reachable from a root vertex \( r \) and the set \anc{r} of vertices that can reach \( r \), there is a decremental data structure that maintains the SCCs and the set of parallel inter-SCC edges of all graphs in \( O(m + n \log n) \) amortized update time, where \( m \) is the number of edges in the current graph \( G \).

Additionally, for all \( r \in V \), the data structure maintains the sets \D{r} and \A{r} of vertices removed from \desc{r} and \anc{r}, respectively, as a consequence of the last update and supports the following operations in time \( O(1) \), where \( X, Y, R \) are the SCCs containing \( x, y, r \), respectively:

\begin{itemize}

    \item 
    \textsc{In\( (y, r) \):}
return \texttt{True} if, in \graph{r}, \( Y \neq R \) and \( Y \) has an in-neighbor from  \( \desc{r} \setminus R \), and \texttt{False} otherwise.

    \item
    \textsc{Out\( (x, r) \):}
return \texttt{True} if, in \graph{r}, \( X \neq R \) and \( X \) has an out-neighbor to  \( \desc{r} \setminus R \), and \texttt{False} otherwise.

\end{itemize}

Lastly, the data structure maintains the set \s{r} of the SCCs in \graph{r} such that \textsc{In\( (y, r) \)} or \textsc{Out\( (x, r) \)} has changed from \texttt{True} to \texttt{False} due to the last deletion.
 
\end{restatable}

We denote by \F{X}{Y} the set of parallel edges in \( G \) between the SCCs \( X, Y \) maintained by \Cref{lem:extended_roditty}.
The lemma below shows how maintaining \desc{u}, \anc{u}, and the parallel edges is useful in maintaining \TR.
Recall that an edge \( \edge{x}{y} \in E \) is redundant if there exists a \( x \to y \) path in \( G \setminus xy \).

\begin{lemma}\label{lem:redundant_general}
Let \( \edge{x}{y} \in E \) be an 
 inter-SCC edge of \( G \), and for a vertex \( z \in V \), let \( X, Y, Z \) be the SCCs containing \( x, y, z \) in \graph{z}, respectively.
Then \edge{x}{y} is redundant in \( G \) iff one of the following holds.
\begin{enumerate}
\item 
Vertex \( z \) exists with \( Z \notin \{ X, Y \} \) such that \( \edge{x}{y} \in \E{z} \), \( x \in \anc{z} \), and \( y \in \desc{z} \) in \graph{z}.

\item 
Vertex \( z \) exists with \( Z = X \) in \graph{z}, such that \( Y \) has an in-neighbor from \( \desc{z} \setminus Z \) in \graph{z}.
\item 
Vertex \( z \) exists with \( Z = Y \) in \graph{z}, such that \( X \) has an out-neighbor to \( \anc{z} \setminus Z \) in \graph{z}.
\item 
There are parallel edges between the SCCs \( X \) and \( Y \) in \( G \).
\end{enumerate}
\end{lemma}
\begin{proof}

We first show the ``if'' direction.
Suppose that \( \edge{x}{y} \) is an inter-SCC redundant edge in \( G \).
Then, there exists a directed path \( P \) from \( x \) to \( y \) in $G \setminus {\edge{x}{y}}$ whose length is at least two. 
Let $z$ be the vertex on \( P \) that has been a center of an insertion most recently. 
At the time of this insertion, all the edges of \( P \) were already present in \graph{z}, and thus \( P \) exists in \graph{z}.
Suppose that (i) \( P \) passes through a SCC different than \( X, Y \) in \( G \).
If \( Z \notin \{X, Y \} \) in \graph{z}, then Item (1) holds.
If \( Z = X \), then Item (2) holds, and if \( Z = Y \), then Item (3) holds.
But if (ii) \( P \) only has vertices in \( X, Y \), then Item (4) holds.

For the ``only-if'' direction, we only prove Item (2); Items (1) and (3) can be shown similarly, and Item (4) is trivial. 
To prove Item (2), let us assume that \( Y \) has an in-neighbor \( w \in  \desc{z} \setminus Z \) in \graph{z}.
Our goal is to show that \( \edge{x}{y} \) is redundant in \graph{z}. 
To this end, since $w \notin Y$ and \( X \neq Y \), there exists a path $P_1$ from $x$ to $w$ with no vertex in \( Y \).
Also, since \( w \) is an incoming vertex to \( Y \), there exists a path \( P_2 \) from \( w \) to \( y \) with no vertex in \( X \).
Therefor, the path \( P \) as the concatenation of \( P_1 \) and \( P_2 \) is a \( x \to w \to y \) path not containing the edge \edge{x}{y}. 
Since \( \E{z} \subseteq E \), \( P \) is also contained in \( G \setminus \edge{x}{y} \), implying that  \edge{x}{y} is a redundant inter-SCC edge in \( G \).
\end{proof}

Note that when no parallel inter-SCC edge exists, an inter-SCC edge \edge{x}{y} belongs to \TR iff it is not redundant, as the latter guarantees  that there is no simple path \( x \to z \to y \) where \( z \) is not belonging to the SCCs containing \( x, y \).
On the other hand, the existence of parallel edges does not change the reachability information nor the size of \TR: if an inter-SCC edge \edge{x}{y} belongs to \TR, then all other parallel edges between \( X \) and \( Y \) cannot belong to \TR.
This motivates us to ``ignore'' parallel inter-SCC edges except one of them in our data structure: we \textit{mark} the front edge in \F{X}{Y} and simply assume that other edges does not belong to \TR.
Therefore, if an inter-SCC edge is not marked, then it does not belong to \TR, and if it is marked and not redundant, then it belongs to \TR.

To incorporate \Cref{lem:redundant_general} in our data structure, we define \counter{xy} and \touch{xy} for every inter-SCC edge \( xy \) as follows, where \( X, Y, Z \) are the SCCs containing \( x, y, z \), respectively, in \graph{z}:
\begin{itemize}
\item 
The counter \counter{\edge{x}{y}} stores the number of vertices \( z \notin \{ x, y \} \) in \graph{z} such that \( Z \notin \{ X, Y \} \), \( \edge{x}{y} \in \E{z} \), \( x \in \anc{z} \), and \( y \in \desc{z} \).

\item
The binary value \touch{\edge{x}{y}} is set to \( 1 \) if there exists a vertex \( z \) such that either \( z \in X \) and \( Y \) has an in-neighbor from \( \desc{x} \setminus X \) in \graph{x} or \( z \in Y \) and \( x \) has an out-neighbor to \( \anc{y} \setminus Y  \) in \graph{y}, and is set to \( 0 \) otherwise.
\end{itemize}

Note that for a redundant edge $\edge{x}{y}$ in $G$, a vertex $z$ contributes towards \counter{\edge{x}{y}} iff $\edge{x}{y} \in \E{z}$ as an inter-SCC edge. 
If there is no such $z$, then \edge{x}{y} has become redundant due the last insertion being centered around \( x \) or \( y \), which in turn implies that \( \touch{\edge{x}{y}} = 1 \).
We conclude the following invariant.
\begin{invariant} \label{inv:general}

An inter-SCC edge \(\edge{x}{y} \in E\) belongs to the transitive reduction \TR iff \edge{x}{y} is a marked parallel edge with \( \counter{\edge{x}{y}} = 0 \) and \( \touch{\edge{x}{y}} = 0 \).

\end{invariant}

In \Cref{subsub:insertion_general,subsub:deletion_general}, we explain how we efficiently maintain \counter{\edge{x}{y}} and \touch{\edge{x}{y}} for every inter-SCC edge \( \edge{x}{y} \in E \).
In the following, we explain how to maintain the intra-SCC edges belonging to \TR.

For intra-SCC edges, we use the following lemma.

\begin{lemma}[\cite{GibbonsKRST91}] \label{lem:minimal}
Given an \( m \)-edge \( n \)-vertex strongly connected graph \( G \), there is an algorithm to compute a minimal strongly connected subgraph of \( G \) in \( O(m + n \log n) \) worst-case time.
\end{lemma}
After each update, we run the algorithm of \Cref{lem:minimal} on all SCCs of \( G \) maintained by the data structure of \Cref{lem:extended_roditty}.
Since this takes \( O(m + n \log n) \) time, this does not affect the update time of our data structure.

In the rest of this section, we focus on maintaining inter-SCC edges as the maintenance of intra-SCC edges is clear from the discussion above.
Each edge will be assumed to be an inter-SCC edge, unless explicitly expressed otherwise.

\subsection{Edge insertions}
\label{subsub:insertion_general}

After the insertion of \Ei{u} centered around $u$, the data structure updates \graph{u}, while leaving every other graph \graph{v} is \emph{unchanged}.

The data structure first computes the sets \desc{u} and \anc{u}, as well as the sets \CC{u} and \B{u} of the vertices added to \desc{u} and \anc{u}, respectively, due to the insertion of \Ei{u}.
It computes the SCCs of \graph{u} while listing the edges entering each SCC \( Z \neq U \) from \( \desc{u} \setminus U \) and the edges leaving the SCC to \( \anc{u} \setminus U \), where \( U \) is the SCC in \graph{u} containing \( u \).
By inspecting \E{u}, the data structure computes the set \F{X}{Y} of parallel edges between each pair \( X, Y \) of SCCs, thus maintaining the marked edges.
This takes \( O(m + n \log n) \) time and reinitializes the data structure of \Cref{lem:extended_roditty} for \graph{u}.

To maintain \counter{xy} for every inter-SCC edge \( xy \in E \), we only need to examine the contribution of \( u \) in \counter{\edge{x}{y}} as only \graph{u} has changed.
By definition, every inter-SCC edge \edge{x}{y} touching the SCC \( U \), i.e., with \( x \in U \) or \( y \in U \), does not contribute in \counter{\edge{x}{y}}.
Note that \E{u} consists of the edges before the update and the newly added edges, and so \( U \) could only grow after the update.
This leads to distinguishing the following cases.
\begin{itemize}
\item[(1)] 
Suppose that \edge{x}{y} is not touching \( U \) and has already existed in \E{u} before the update.
Then, \counter{\edge{x}{y}} will increase by one only if the update makes \( x \) to reach \( y \) through \( u \). 
i.e., there is no path \( x \to u \to y \) before the update  (\( x \notin \anc{u} \setminus \B{u} \) or \( y \notin \desc{u}\setminus \CC{u} \)), but there is at least one afterwards (i.e., \( x \in \anc{u} \) and \( y \in \desc{u} \)).

\item[(2)]
Suppose that \edge{x}{y} is not touching \( U \) and has added to \E{u} after the update.
Since this is the first time we examine the contribution of \( u \) towards \counter{\edge{x}{y}}, we increment \counter{\edge{x}{y}} by one if \( x \) can reach \( y \) through \( u \) (i.e., \( x \in \anc{u} \) and \( y \in \desc{u} \)).

\item[(3)]
Suppose that \edge{x}{y} is a newly edge touched by \( U \).
i.e., \( xy \) touching \( U \) as a result of the growth of \( U \) after the update.
Then, \counter{\edge{x}{y}} will decrease by one only if the update makes \( x \) could reach \( y \) through \( u \) before the update (\( x \in \anc{u} \setminus \B{u} \) and \( y \in \desc{u} \setminus \CC{u} \)). 
\end{itemize}

Cases (1) and (2) are handled by inspecting the inter-SCC edges touching \( \B{u} \cup \CC{u} \).
To prevent double-counting, a similar partitioning in \Cref{subsub:deletion_dag} can be used.
For case (3), we can afford to inspect all inter-SCC edges touching the vertices added to \( U \) due to the update.

To maintain \touch{xy} for every inter-SCC edge \( xy \in E \), we distinguish the two following cases.
Since only \graph{u} changes, by definition, \( \touch{\edge{x}{y}} \) could be affected only if \edge{x}{y} touches \( U \).
As the update on \graph{u} adds more edges to it, \touch{\edge{x}{y}} can only change from $0$ to~$1$.
For every edge \( \edge{x}{y} \in E \) touching \( U \), we set $\touch{\edge{x}{y}} \gets 1$ if one of the following happen:
\begin{itemize}
\item[(a)] 
$x \in U$ and \( Y \) has an in-neighbor \( w \notin U \) in \graph{u} reachable from \( u \) (i.e., if $\textsc{In}(y, x)$ reports $\texttt{True}$), or

\item[(b)]
$y \in U$ and \( X \) has an out-neighbor \( w \notin U \)  in \graph{u} that can reach \( Y \) (i.e., if $\textsc{Out}(x, y)$ reports $\texttt{True}$).
\end{itemize}
Note that both cases above try to insure the existence of a path \( x \to w \to y \) when \( x \in U \) or \( y \in U \).

\begin{lemma} \label{lem:insertion_general}
After each insertion, the data structure maintains the transitive reduction \TR of \( G \) in \( O(m + n\log n) \) worst-case time, where \( m \) is the number of edges in the current graph \( G \).
\end{lemma}
\begin{proof}
Suppose that the insertion of edges \Ei{u} has happened centered around \( u \).

\underline{Correctness:}
For \counter{xy}, the proof of correctness for the case \( U \notin \{ X, Y\} \) is similar to that in the proof of \Cref{lem:insertion}, and is thus omitted.
Suppose that \( U = X \) or \( U = Y \).
If \( u \) has been counted in \counter{\edge{x}{y}}, then we must decrement \counter{\edge{x}{y}} because \( u \) no longer satisfies the conditions to be included.
i.e., there was a \( x \rightarrow u \rightarrow y \) path in \graph{u} before the update, implying that \( x \) and \( y \) already belonged to \anc{u} and \desc{u}, respectively.
This means that \( x \in \anc{u} \setminus \B{u} \) and \( y \in \desc{u} \setminus \CC{u} \).

For \touch{xy}, by construction, it follows that the values \textsc{Out\( (y, u) \)} and  \textsc{In\( (x, u) \)} are maintained correctly after the update.
The correctness then follows immediately.

\underline{Update time:} 
As explained before, for each edge \( xy \in \E{u} \), we can check in \( O(1) \) whether we need to update \counter{xy} or \touch{xy} due to the last insertion around \( u \).
Thus, the update time is dominated by (i) the time required for reinitializing \graph{u} and the required reachability information. and (ii) the time required for maintaining SCCs, and (iii) the time required for maintaining intra-SCC minimal subgraphs.
By a similar discussion from \Cref{subsec:combinatorial_dag}, \Cref{th:roditty_SCC}, and \Cref{lem:minimal}, it follows that (i), (ii), and (iii) take \( O(m + n \log n) \), respectively.
\end{proof}

\subsection{Edge deletions}
\label{subsub:deletion_general}

After the deletion of \Ed, the data structure passes the deletion to \emph{every} \graph{u}, \( u \in V \).
Let \D{u} and \A{u} denote the sets of vertices removed from \desc{u} and \anc{u}, respectively, due to the deletion of \Ed.

We decrementally maintain \desc{u} and \anc{u}, and the sets \D{u} and \A{u} using the data structure of \Cref{lem:extended_roditty}, which is initialized last time \graph{u} was rebuilt due to an insertion centered around \( u \). 
Since the data structure maintains the set \F{X}{Y} of parallel edges between each pair \( X, Y \) of SCCs, the marked inter-SCC edges are automatically maintained.

To maintain \counter{\edge{x}{y}} for any inter-SCC edge \( \edge{x}{y} \in E \), we need to cancel out every vertex \( z \notin \{ x, y \} \) that contained a path \( x \to z \to y \) in \graph{z} before the update but no longer has one.
i.e., \( x \in \anc{z} \cup \A{z} \) and \( y \in \desc{z} \cup \D{z} \) in \graph{z} and either \( x \in \A{z} \) or \( y \in \D{z} \).
To avoid double-counting \( u \) for the edges with both endpoints in \( \A{u} \cup \D{u} \), we use a similar approach explained in \Cref{subsub:deletion_dag}: partition \( \A{u} \cup \D{u} \) into three disjoint sets and argue how to handle each.
Contrary to DAGs, there may be new inter-SCC edges form in \graph{u} due to SCC splits, which are given by the data structure of \Cref{lem:extended_roditty}.
For each such edge \edge{x}{y}, we can check in constant time if \( x, y \notin U \), and if they satisfy the conditions of \counter{\edge{x}{y}}, we increase \counter{\edge{x}{y}} by one.

To maintain \touch{xy} for an inter-SCC edge \( xy \in E \), we distinguish the following three cases.
Let \( X, Y, U \) be the SCCs containing \( x, y, u \) in \graph{u}, respectively.
\begin{enumerate}
    \item[(a)] 
    Vertex \( x \) or \( y \) has detached from \( U \) due to the update.
    In this case, if \( \textsc{In\( (u,y) \)} = \texttt{True} \) or \( \textsc{Out\( (u,x) \)} = \texttt{True} \) in \graph{u} before the update, we set \( \touch{\edge{x}{y}} \gets \touch{\edge{x}{y}} - 1 \).

    \item[(b)]
    Vertices \( x \) and \( y \) has not detached from \( U \), but there is no longer a \( x \rightarrow y \) path passing through an SCC different than \( X, Y \) in \graph{u}.
    In this case, we only need to inspect the edges touching the SCCs in \s{u}: if \( \textsc{In\( (u,y) \)} = \texttt{True} \) or \( \textsc{Out\( (u,x) \)} = \texttt{True} \) in \graph{u} before the update, we set \( \touch{\edge{x}{y}} \gets \touch{\edge{x}{y}} - 1 \). 
    
    \item[(c)]
    If \edge{x}{y} is a new inter-SCC edge in \graph{u}.
    In this case, we only need to inspect the new inter-SCC edges touching \( U \) after the update, and if \( \textsc{In\( (u,y) \)} = \texttt{True} \) or  \( \textsc{Out\( (u,x) \)} = \texttt{True} \), we set \( \touch{\edge{x}{y}} \gets \touch{\edge{x}{y}} + 1 \).
    
\end{enumerate}

\begin{lemma} \label{lem:deletion_genral}
After each deletion, the data structure maintains the transitive reduction \TR of \( G \) in \( O(m + n \log n) \) amortized update time, where \( m \) is the number of edges in the current graph \( G \).
\end{lemma}
\begin{proof}
Assume that the deletion of \Ed has happened.

\underline{Correctness:}
We prove that the values \counter{\cdot} and \touch{\cdot}
are maintained correctly.

For \counter{\cdot}, the proof of correctness is similar to that in the proof of \Cref{lem:deletion}.

For \touch{\cdot}, note that only one of \( \textsc{In\( (u,y) \)} = \texttt{True} \) or \( \textsc{Out\( (u,x) \)} = \texttt{True} \) can hold for the case (a).
This is because either \( x \in U \) or \( y \in U \).
We can simply have access to the values in \graph{u} before the update, by making a copy of it before updating \graph{u}.
The details of the proof for the cases (a) and (b) are then similar to that in the proof of \Cref{lem:deletion}.
For the case (c), note that, by the definition of \touch{\cdot}, a new inter-SCC edge \edge{x}{y} satisfies the condition iff \( \textsc{In\( (u,y) \)} = \texttt{True} \) or \( \textsc{Out\( (u,x) \)} \).
This completes the correctness for \touch{\cdot}.

\underline{Update time:}
Since the graphs are decremental, it follows from \Cref{lem:extended_roditty} that SCCs and the required reachability information for all graphs are maintained in \( O(m + n \log n) \).
We explain the time required for maintaining \counter{xy} and \touch{xy} for every edge \( xy \in E \).

for maintaining \counter{xy}, since \graph{u} is decremental, each vertex appears in \( \A{u} \cup \D{u} \) at most twice, and each edge becomes an inter-SCC edge at most once.
Thus, the total time to maintain \counter{\cdot} for any sequence of edge deletion is \( O(mn) \).

For maintaining \touch{xy}, note that we can efficiently maintain (a), (b), and (c) since the new SCCs and inter-SCC edges are already computed by the data structure of \Cref{lem:extended_roditty}.
Since \graph{u} is decremental, we inspect \( O(m) \) edges in total after any sequence of edge deletions.
Thus, the total time to maintain \touch{\cdot} for any sequence of edge deletion on \( G \) is \( O(mn) \).
\end{proof}

\subsection{Space complexity}
\label{subsubsec:space_general}

A similar encoding explained in \Cref{subsubsec:space} results in \( O(n^2) \) space.

}

\section{Algebraic Dynamic Transitive Reduction} \label{sec:algebraic}

\arxivVsConference{
We start by explaining the data structure for DAGs in \Cref{subsec:algebraic-dag},
and then present the general case in \Cref{sec:app_algebraic_general}.
}{}

{
\let\section\subsection
\let\subsection\subsubsection

\newcommand{\Ot}{\ensuremath{\widetilde{O}}}
\newcommand{\eps}{\ensuremath{\epsilon}}
\newcommand{\field}{\mathbb{F}}

\section{Algebraic Dynamic Transitive Reduction on DAGs} \label{subsec:algebraic-dag}
In this section we give algebraic dynamic algorithms for transitive reduction
in DAGs\arxivVsConference{}{ (for general digraphs, see \Cref{sec:app_algebraic_general})}.
We reduce the problem to maintaining the inverse
of a matrix associated with $G$ and (in the general case) testing some identities involving the
elements of the matrix.

We will need the following results on dynamic matrix inverse maintenance. 

\begin{theorem}\label{thm:dyn-inv-rows}{\upshape\cite{Sankowski04}}
Let $A$ be an $n\times n$ matrix over a field $\field$ that is invertible at all times.
    There is a data structure maintaining $A^{-1}$ explicitly subject to
    row or column updates issued to $A$ in $O(n^2)$ worst-case update time. The data structure is initialized in $O(n^\omega)$ time and uses $O(n^2)$~space.
\end{theorem}

\begin{theorem}\label{thm:dyn-inv-elem}{\upshape\cite{Sankowski04, DBLP:conf/cocoon/Sankowski05}}
Let $A$ be an $n\times n$ matrix over a field $\field$ that is invertible at all times, $a\in (0,1)$, and $Y$ be a (dynamic) subset
    of $[n]^2$.
    There is a data structure maintaining $A^{-1}$ subject to \emph{single-element}
    updates to $A$ that maintains $A^{-1}_{i,j}$ for all $(i,j)\in Y$
    in $O(n^{\omega(1,a,1)-a}+n^{1+a}+|Y|)$ worst-case time per update. The data structure additionally supports:
    \begin{enumerate}
        \item square submatrix queries $A^{-1}[I,I]$, for $I\subseteq [n]$, $|I|=n^{\delta}$ in $O(n^{\omega(\delta,a,\delta)})$ time,
        \item given $A^{-1}_{i,j}$, adding or removing $(i,j)$ from $Y$,
        in $O(1)$ time.
    \end{enumerate}
    The data structure can be initialized in $O(n^\omega)$ time and uses $O(n^2)$ space.
\end{theorem}
Above, $\omega(p,q,r)$ denotes the exponent such that multiplying $n^p\times n^q$ and
$n^q\times n^r$ matrices over~$\field$ requires $O(n^{\omega(p,q,r)})$ field operations.
Moreover $\omega:=\omega(1,1,1)$.

For DAGs, efficient algebraic fully dynamic transitive reduction algorithms follow from \Cref{thm:dyn-inv-rows,thm:dyn-inv-elem} rather easily by applying the reduction of \Cref{l:dag-reduction}. 
To show that, we now recall how the 
algebraic dynamic transitive closure data structures for DAGs~\cite{DemetrescuI05, DBLP:journals/jcss/KingS02}
are obtained.

Identify $V$ with $[n]$. Let $A(G)$ be the standard adjacency matrix of $G=(V,E)$, that is,
for any $u,v\in V$, $A(G)_{u,v}=1$ iff $uv\in E$, and $A(G)_{u,v}=0$ otherwise.
It is well-known that the powers of $A$ encode
the numbers of walks between specific endpoints in $G$. That is, for any $u,v\in V$
and $k\geq 0$, $A(G)_{u,v}^k$ equals the number of $u\to v$ $k$-edge walks in $G$.
In DAGs, all walks are actually paths. Moreover, we have the following property:
\begin{lemma}\label{lem:dag-path-count}
    If $G$ is a DAG, then the matrix $I-A(G)$ is invertible. Moreover, $(I-A(G))^{-1}_{u,v}$
    equals the number of paths from $u\to v$ in $G$.
\end{lemma}
\begin{proof}
Put $A:=A(G)$. If $G$ is a DAG, then $A^n$ is a zero matrix since every $n$-edge walk has to contain a cycle. From that it follows that $(I-A)(I+A+\ldots+A^{n-1})=I$, i.e., $\sum_{i=0}^{n-1}A^i$
is the inverse of $I-A$. On the other hand, the former matrix clearly encodes path counts of all the possible lengths $0,1,\ldots,n-1$ in $G$.
\end{proof}
By the above lemma, to check whether a path $u\to v$ exists in a DAG $G$, it is enough to test whether $(I-A(G))^{-1}_{u,v}\neq 0$.
This reduces dynamic transitive closure on $G$ to maintaining the inverse of $I-A(G)$ dynamically.
There are two potential problems, though: (1) (unbounded) integers do not form a field (and \Cref{thm:dyn-inv-rows,thm:dyn-inv-elem} are stated for a field), (2) the path counts in $G$ may be very large integers of up to $\Theta(n)$
bits, which could lead to an $\Ot(n)$ bound for performing arithmetic operations while maintaining the path counts. A standard way to address both problems (with high probability $1-1/\poly{n}$) is to perform all the counting modulo a sufficiently large random prime number~$p$ polynomial in $n$ (see, e.g.,~\cite[Section~3.4]{DBLP:journals/jcss/KingS02} for an analysis). Working over $\mathbb{Z}/p\mathbb{Z}$ solves both problems as arithmetic operations modulo $p$ can be performed in $O(1)$ time on the word RAM.

\begin{theorem}\label{thm:algebraic-dag}
    Let $G$ be a fully dynamic DAG. The transitive reduction of $G$ can be maintained:
    \begin{enumerate}[label=(\arabic*)]
        \item in $O(n^2)$ worst-case time per update if vertex updates are allowed,
        \item in $O(n^{1.528}+m)$ worst-case time per update if only single-edge updates are supported.
    \end{enumerate}
    Both data structures are Monte Carlo randomized and give correct outputs with high probability.
    They can be initialized in $O(n^\omega)$ time and use $O(n^2)$ space.
\end{theorem}
\begin{proof}
\Cref{l:dag-reduction} reduces maintaining the set of redundant edges
of $G$ to a dynamic transitive closure
data structure maintaining reachability for $O(m)$ pairs $Y\subseteq A\times A$, supporting
either:
\begin{enumerate}[label=(\arabic*)]
    \item vertex-centered updates to both the underlying graph and the set $Y$, or
    \item single-edge updates to the underlying graph and single-element updates to $Y$.
\end{enumerate}
By our discussion, the former data structure can be obtained by
applying \Cref{thm:dyn-inv-rows}
to the matrix $I-A(G)$.
This way, for vertex updates, one obtains $O(n^2)$ worst-case update time immediately.

The latter data structure is obtained by applying \Cref{thm:dyn-inv-elem}
with $a\approx 0.528$ that satisfies $\omega(1,a,1)=1+2a$ to the matrix $I-A(G)$. Then, a single-edge update is handled in $O(n^{1+a}+m)$ time. Note that inserting a new element to $Y$ in \Cref{thm:dyn-inv-elem} requires
computing the corresponding element of the inverse. This requires $O(n^{\omega(0,a,0)})=O(n^a)$ extra time and is negligible compared to the $O(n^{1+a})$ cost of the element update.
\end{proof}

\section{Algebraic Dynamic Transitive Reduction on General Graphs}
\label{sec:app_algebraic_general}

In this section we give algebraic dynamic algorithms for general digraphs.

In general graphs, simple path counting fails since there can be infinite
numbers of walks between vertices. This is why the algebraic dynamic transitive closure
data structures dealing with general graphs~\cite{Sankowski04} require
more subtle arguments. For maintaining the transitive reduction, we will
rely on those as well. Identifying redundant edges -- specifically groups of parallel
redundant inter-SCC edges -- will also turn out more challenging.
In fact, the obtained data structure for single-edge updates will not be as efficient
as the corresponding data structure for DAGs (for the current fast matrix multiplication exponents).

In this section, we will assume that $G=(V,E)$ has \emph{a self loop on every vertex},
i.e., $uu\in E$ for all $u\in V$, even though loops are clearly
redundant from the point of view of reachability.
For this reason, we also assume the updates only add or delete non-loop edges.
When reporting the transitive reduction, we do not report the $n$ self-loop edges
that are only helpful internally.

\newcommand{\asymb}{\tilde{A}}

\subsection{The data structure}
For general graphs we will identify redundant intra-SCC and inter-SCC edges separately.
Upon update issued to $G$, the first step is always to compute the strongly-connected
components of~$G$ from scratch in $O(m)$ time.
This yields the partition of edges $E$ into inter- and intra-SCC, and also sets of \emph{parallel} inter-SCC edges, i.e., those that connect
the same pair of SCCs.

The intra-SCC edges can be handled in $O(m+n\log{n})$ time worst-case time per update using~\Cref{lem:minimal} as was done in~\Cref{subsec:combinatorial_general}.
In the following, we focus on handling the inter-SCC edges.

Fix a group $F$ of $k$ parallel inter-SCC edges $\{u_iv_i:i=1,\ldots,k\}$,
where
$\{u_1,\ldots,u_k\}$ are in the same SCC $R$ of $G$
and $\{v_1,\ldots,v_k\}$ are in the same SCC $T$ of $G$,
and $R\neq T$.
Recall that $F$ is redundant if \emph{all} edges in $F$ can be removed without changing the transitive closure of~$G$. %
If $F$ is non-redundant, we need to keep exactly one, arbitrarily chosen
edge of $F$ in the transitive reduction: in such a case, after contracting the SCCs and removing parallel edges, $F$ corresponds to a redundant edge of the obtained DAG.
Therefore, given $F$, we only need to decide whether this group is redundant. 

Let $r,t$ be arbitrarily chosen vertices of $R,T$ respectively.
Note that~$F$ is redundant if and only if there exists
a path $r\to t$ going through a vertex in some strongly connected component $Y\notin \{R,T\}$ of $G$.
Our strategy will be to express this property of $F$ algebraically.
To this end, we need to introduce some more notions.

Let $X=\{x_{u,v}:u,v\in V\}$ be a set of $n^2$ formal variables associated with
pairs of vertices of $G$.
Fix some finite $\field$. Denote by $\field(X)$ the field of multivariate rational functions
with coefficients from $\field$ and variables from $X$.
Following~\cite{Sankowski04}, let us define a \emph{symbolic adjacency matrix} $\asymb(G)\in (\field(X))^{n\times n}$ of $G$
to be a matrix such
that
\begin{equation*}
\asymb(G)_{u,v} =
 \left\{
 \begin{array}{rl}
   x_{u,v} & \textrm{if  } uv \in E, \\
 0 & \textrm{otherwise.}
\end{array}
\right.
\end{equation*}

The self loops in $G$ guarantee that $\asymb(G)$ is invertible over $\field(X)$, i.e., $\det(\asymb(G))\neq 0$~\cite{Sankowski04}.
In \Cref{s:redundant-algebraic} we are going to prove the following theorem.
\begin{theorem}\label{t:tr-matrix}
Let $F=\{u_iv_i:i=1,\ldots,k\}$ be a group of parallel inter-SCC edges between SCCs $R,T$ of $G$.
Let $r,t$ be arbitrary vertices of $R,T$, resp.
  The group $F$ is redundant iff:
\begin{align*}
  \asymb(G)^{-1}_{r,t}\not\equiv-\sum_{i=1}^k (-1)^{u_i+v_i}\cdot x_{u_i,v_i}\cdot \asymb(G)^{-1}_{r,u_i}\cdot \asymb(G)^{-1}_{v_i,t}.
\end{align*}
\end{theorem}
Let us now show how \Cref{t:tr-matrix} can be used to obtain a dynamic transitive reduction data structure.
It is enough to ``maintain'' the inverse
of the matrix $\asymb(G)$ and, after each update, check appropriate polynomial identities
on some $m$ entries of the inverse.
The inverse's entries are rational functions, so maintaining them would be too slow.
Instead, we
maintain the entries' evaluations for
some uniformly random substitution 
$\bar{X}:X\to \mathbb{Z}/p\mathbb{Z}$, where $\bar{X}=\{\bar{x}_{u,v}:u,v\in V\}$,
and $\field=\mathbb{Z}/p\mathbb{Z}$,
where $p$ is a prime number $p=\Theta(n^{3+c})$, $c>0$.

\begin{lemma}\label{l:random-substitution}
Let $\field$ and
$A\in \field^{n\times n}$ be a obtained as described above.
Then, with probability at least $1-O(1/n^{2+c})$, $A$ is invertible and $F$ (as defined in \Cref{t:tr-matrix}) is redundant iff:
\begin{equation*}
A^{-1}_{r,t}+\sum_{i=1}^k (-1)^{u_i+v_i}\cdot \bar{x}_{u_i,v_i}\cdot A^{-1}_{r,u_i}\cdot A^{-1}_{v_i,t}\neq 0.
\end{equation*}
\end{lemma}
\begin{proof}
First of all, by the Schwartz-Zippel lemma, since $\det(\asymb(G))$ is not a zero polynomial (by the non-singularity of $\asymb(G)$), $\det(A)= 0$
with probability at most:
\begin{equation*}
\frac{\deg(\det(\asymb(G)))}{|\field|}\leq \frac{n}{p}=O(1/n^{2+c}).
\end{equation*}
Now suppose $\det(A)\neq 0$. Consider the polynomial:
\begin{equation*}
P_F(X)=\asymb(G)^{-1}_{r,t}+\sum_{i=1}^k (-1)^{u_i+v_i}\cdot x_{u_i,v_i}\cdot \asymb(G)^{-1}_{r,u_i}\cdot \asymb(G)^{-1}_{v_i,t}.
\end{equation*}
By \Cref{t:tr-matrix}, if $F$ is non-redundant, then $P_F(X)\equiv 0$. So for any substitution $\bar{X}$, $P_F(\bar{X})=0$. 

Now assume $F$ is redundant.
Then $P_F(X)\not\equiv 0$ by \Cref{t:tr-matrix}. By the relationship between the adjoint and the inverse of $\asymb(G)$,
the elements of $\asymb(G)^{-1}$
are rational functions whose numerators are polynomials of degree at most $n$, and the denominator is a divisor of $\det(\asymb(G))$.
As a result, $P_F(X)$ is a rational function 
whose numerator is a polynomial of degree no more than $2n+1$.
By the Schwartz-Zippel lemma, $P_F(\bar{X})=0$ can happen
with probability at most $(2n+1)/|\field|=O(1/n^{2+c})$.
\end{proof}
As $G$ has at most $n^2$ edges, by the union bound we get that the identity in \Cref{l:random-substitution} can be used to correctly classify the inter-SCC edges of $G$ as either redundant or non-redundant with probability at least $1-O(1/n^c)$.
The data structure simply maintains (a part of) the inverse of $A$ and checks the identities
of \Cref{l:random-substitution} after each update.
\emph{Assuming all the needed elements
of $A^{-1}$ are computed}, testing whether the group $F$ is redundant
takes $O(|F|)$ time. As a result, through all the (groups of) edges of $G$, such a check requires $O(m)$ time.
The elements of $A^{-1}$ are never revealed to the adversary, so a single variable substitution picked at the beginning can be used over many updates. By amplifying the constant $c$, we can guarantee high-probability correctness over a polynomial number of updates to $G$.

\subsection{Maintaining the required elements of the inverse $A^{-1}$.}
If we want to support vertex updates changing all edges incident to a single vertex, then we can maintain the entire $A^{-1}$
in $O(n^2)$ worst-case time per update
using \Cref{thm:dyn-inv-rows}. In such a case, the worst-case update time of our dynamic transitive reduction data structure is $O(n^2)$.

Let us now consider more interesting single-edge updates when $G$ is not dense.
Recall that in the acyclic case, it was sufficient
to maintain $m$ entries of the inverse corresponding precisely to the edges of $G$.
In the general case, by \Cref{l:random-substitution},
deciding a group of $k$ parallel
inter-SCC edges requires inspecting $2k$
elements of the inverse.
However, some of them are not in a 1-1 correspondence with the parallel edges of that group; they depend
on the ``rooting'' $(r,t)$ (of our choice) of the SCCs $R$ and $T$, respectively.
Unfortunately, the groups of parallel edges can change quite dramatically and unpredictably upon updates, e.g., if a group splits as a result of an edge deletion,
the chosen rootings might not be helpful at deciding whether the groups after the split are redundant.

We nevertheless obtain a bound close to that we got for DAGs by applying a heavy/light distinction to the SCCs and exploiting randomization further.
Let $\delta\in [0,1]$ be a parameter to be fixed later.
We call an SCC \emph{heavy} if it has at least $n^\delta$ vertices,
and \emph{light} otherwise.

A well-known fact~\cite{UY91} says
that if we sample a set $S$ of $\Theta(n^{1-\delta}\log{n})$ random vertices of~$G$, then any
subset of $V$ of size at least $n^\delta$ (chosen independently of $S$) will contain
a vertex of $S$ with high probability (dependent on the constant hidden in the $\Theta$ notation). 
In particular, one can guarantee
that $S$ \emph{hits} (w.h.p.) every
out of $\poly{n}$ subsets of $V$
chosen independently of $S$.
In our data structure, 
we sample one such \emph{hitting set} $S$ of
size $\Theta(n^{1-\delta}\log{n})$
so that it hits all the heavy SCCs
that ever arise with high probability. For that, we will never leak $S$ to the adversary (unless the data structure errs, which happens will low probability) so that $S$ remains independent of the current structure of the SCCs.
Next, we employ the dynamic matrix inverse data structure $\mathcal{D}$ of \Cref{thm:dyn-inv-elem}. We define the ``set of interest'' $Y$ in that data structure, to contain at all times:
\begin{enumerate}
\item $(S\times V)\cup (V\times S)$,
\item $(u,v)$ for all $uv\in E$,
\item $B\times B$ for every \emph{light} SCC $B$ of $G$.
\end{enumerate}
Note that the size of $Y$ is at most
\begin{equation*}
m+2n|S|+\sum_{\substack{B\subseteq V\\B\text{ is a light SCC}}}|B|^2\leq m+
\Ot(n^{2-\delta})+n^\delta\left(\sum_{\substack{B\subseteq V\\B\text{ is a light SCC}}}|B|\right)\leq m+\Ot(n^{2-\delta})+n^{1+\delta}.
\end{equation*}
\begin{lemma}
    After a single-edge update issued to $G$, the data structure $\mathcal{D}$ can be updated in $\Ot(n^{\omega(1,a,1)-a}+n^{1+a}+n^{1-\delta+\omega(\delta,a,\delta)}+m+n^{2-\delta})$
    worst-case time.
\end{lemma}
\begin{proof}
Recall that the worst-case update time of $\mathcal{D}$ is $O(n^{\omega(1,a,a)-a}+n^{1+a}+|Y|)$ and we have
$|Y|\leq m+n^{2-\delta}+n^{1+\delta}$. It is enough to argue that computing the
new elements
that appear in $Y$ after an update
takes $O(n^{1-\delta+\omega(\delta,a,\delta)})$ additional time.
Note that $O(n^{1+\delta})\subseteq O(n^{1-\delta+\omega(\delta,a,\delta)})$.

First of all, we may need to add/remove from $Y$ a single element corresponding to the updated edge. It can also happen
that after the update, a number
of new light SCCs $B_1,\ldots,B_\ell$
arise.\footnote{After an insertion, we necessarily have $\ell=1$. However, a deletion may cause a very large SCC of size $\Theta(n)$ split into many small SCCs.} In such a case, we query the
data structure $\mathcal{D}$ for the $\ell$ submatrices $B_i\times B_i$ of the inverse. If $|B_i|=n^{\delta_i}$, then this costs $O(Q)$ time, where $Q:=\sum_{i=1}^\ell n^{\omega(\delta_i,a,\delta_i)}$ time.

To bound $Q$, note that
$n^{\omega(\delta_i,a,\delta_i)}+n^{\omega(\delta_j,a,\delta_j)}\leq n^{\omega(q,a,q)}$,
where $n^q=n^{\delta_i}+n^{\delta_j}$ by the definition
of matrix multiplication.
As a result, while there exist
some two $\delta_i,\delta_j$ such that $n^{\delta_i}+n^{\delta_j}\leq n^{\delta}$,
we can replace $n^{\delta_i},n^{\delta_j}$ with $n^q=n^{\delta_i}+n^{\delta_j}$
without changing $\sum n^{\delta_i}$  so that the quantity $Q$ does not decrease.
At the end of this process, all $n^{\delta_i}$, except perhaps a single one, 
will be larger than $\frac{1}{2}n^{\delta}$.
As a result, their number $\ell$ will be at most $2n^{1-\delta}+1$. Consequently,
the sum $Q$ obtained after the transformations can be bounded by $3n^{1-\delta+\omega(\delta,a,\delta)}$.
It follows that the quantity $Q$ at the beginning of the process was also
bounded by that.
\end{proof}
Given the elements $Y$ of $A^{-1}$, let us now prove that we can 
test the identities in \Cref{l:random-substitution} in $O(m)$ time.
Consider a group $F$ of parallel inter-SCC edges $\{u_1v_1,\ldots,u_kv_k\}$
so that $u_i\in R$ and $v_i\in T$, where $R,T$ are distinct SCCs of $G$.
We now fix the ``rooting'' of $R$ and $T$ as follows.
If $R$ is heavy, then we pick $r\in R\cap S$, and otherwise, we put $r=u_1$.
Similarly, if $T$ is heavy, we pick $t\in T\cap S$, and otherwise, we put $t=v_1$.
This way, we have $(r,u)\in Y$ for any $u\in R$, since either\
$R$ is light and $R\times R\subseteq Y$, or $R$ is heavy and $\{r\}\times R\subseteq S\times V\subseteq Y$.
Similarly, one can argue that $(v,t)\in Y$ for any $v\in T$.
It follows that all the elements of the form $A^{-1}_{r,u_i}$
or $A^{-1}_{v_i,t}$ required by \Cref{l:random-substitution} can be accessed in $O(1)$ time. And so can be the element $A^{-1}_{r,t}$ since $(r,t)\in (S\times V)\cup (V\times S)$
if either $S$ or $T$ is heavy, and we have $(r,s)=(u_1,v_1)$ otherwise,
which implies $rs\in E$ and thus $(r,s)\in Y$ as well.
We obtain the following:
\begin{theorem}\label{thm:algebraic-general}
    Let $G$ be fully dynamic. The transitive reduction of $G$ can be maintained:
    \begin{enumerate}[label=(\arabic*)]
        \item in $O(n^2)$ worst-case time per update if vertex updates are allowed,
        \item for any $a,\delta\in [0,1]$, in $O(n^{\omega(1,a,1)-a}+n^{1+a}+n^{1-\delta+\omega(\delta,a,\delta)}+n^{2-\delta}+m)$ worst-case time per update if only single-edge updates are supported. In particular, for $a=0.4345$
        and $\delta=0.415$, this update bound is $O(n^{1.585}+m)$.\footnote{This choice of parameters $a,d$ can be obtained using the online complexity term balancer~\cite{Complexity}.%
        }
    \end{enumerate}
    Both data structures are Monte Carlo randomized and give correct outputs with high probability.
    They can be initialized in $O(n^\omega)$ time and use $O(n^2)$ space.
\end{theorem}

\subsection{Proof of \Cref{t:tr-matrix}}\label{s:redundant-algebraic}
In this section we prove \Cref{t:tr-matrix}
by studying cycle covers; this approach has also been used by~\cite{Sankowski04}.

Let a \emph{cycle cover} of $G$ be a collection of vertex-disjoint simple cycles $\mathcal{C}=\{C_1,\ldots,C_k\}$
such that $C_i\subseteq G$ and $\sum_{i=1}^k|C_i|=n$.
Denote by $K(G)$ the set of cycle covers of $G$.
Any simple cycle $C\subseteq G$, $C=v_1v_2\ldots v_\ell$ has an associated monomial
$$\mu(C)=-\prod_{i=1}^\ell -x_{v_i,v_{i+1}}=(-1)^{\ell+1}\prod_{i=1}^\ell x_{v_i,v_{i+1}},$$
where $v_{\ell+1}:=v_1$. For a cycle cover $\mathcal{C}\in K(G)$, we define $\mu(\mathcal{C})=\prod_{C\in\mathcal{C}} \mu(C)$.

Now, by the Leibniz formula for the determinant, we have:
\begin{fact}\label{f:cycle-cover}
  $\det(\asymb(G))=\sum_{\mathcal{C}\in K(G)} \mu(\mathcal{C})$.
\end{fact}
As a result, $\asymb(G)$ is invertible, since the polynomial $\det(\asymb(G))$
contains the monomial $\prod_{u\in V}x_{u,u}$
corresponding to the trivial cycle cover of $G$ that consists of
self-loops exclusively.

\begin{lemma}\label{l:det-scc}
  Let $S_1,\ldots,S_s\subseteq V$ be the strongly connected components of $G$. 
  Then:
  \begin{equation*}
      \det(\asymb(G))=\prod_{i=1}^s \det(\asymb(G[S_i])).
  \end{equation*}
\end{lemma}
\begin{proof}
  Note that each cycle $C$ of a cycle cover $\mathcal{C}\in K(G)$ is contained in precisely
  one SCC of $G$. Hence, there is a 1-1 correspondence between the set of cycle
  covers of $G$ and the $s$-tuples of cycle covers of the individual SCCs of $G$.
  So by \Cref{f:cycle-cover}, we have:
  \begin{align*}
    \det(\asymb(G))&=\sum_{\mathcal{C}\in K(G)} \mu(\mathcal{C})\\
                   &=\sum_{\mathcal{C}_1\in K(G[S_1])}\cdots \sum_{\mathcal{C}_s\in K(G[S_s])}\mu(\mathcal{C}_1)\cdots\mu(\mathcal{C}_s)\\
                   &=\left(\sum_{\mathcal{C}_1\in K(G[S_1])}\mu(\mathcal{C}_1)\right)\cdots \left(\sum_{\mathcal{C}_s\in K(G[S_s])}\mu(\mathcal{C}_s)\right)\\
                   &=\prod_{i=1}^s \det(\asymb(G[S_i])). \qedhere
  \end{align*}
\end{proof}
\begin{corollary}\label{l:det-scc-split}
Let $U\subseteq V$ be such that each SCC of $G$
  is fully contained in either $U$ or in $V\setminus U$.
  Then, 
  \begin{equation*}
      \det(\asymb(G))=\det(\asymb(G[U]))\cdot \det(\asymb(G[V\setminus U])).
  \end{equation*}
\end{corollary}

For $u,v\in V(G)$, let $\asymb^{v,u}(G)$ be obtained from $\tilde{A}(G)$ by zeroing all entries
in the $v$-th row and $u$-th column of $\asymb(G)$ and setting the entry $(v,u)$ to $1$.

Denote by $\mathcal{P}_{u,v}(G)$ the set of all simple $u\to v$ paths in $G$. We extend $\mu$ to paths and
set $\mu(P)=\prod_{i=1}^{\ell-1}-x_{v_i,v_{i+1}}$ if $P=v_1\ldots v_\ell$.

\begin{lemma}\label{l:sumpath}
  For any $v,u\in V$, we have:
  $$\det(\asymb^{v,u}(G))=\sum_{P\in \mathcal{P}_{u,v}} \mu(P)\cdot \det(\asymb(G[V\setminus V(P)])).$$
\end{lemma}
\begin{proof}
  For any cycle $C$, let 
  $$\mu'(C)=\begin{cases}\mu(C) & \text{if } \{u,v\}\cap V(C)=\emptyset \\0 & \text{if } \{u,v\}\cap V(C)\neq \emptyset \text{ and } vu\notin E(C),\\
  \mu(C)/x_{v,u} & \text{if } vu\in E(C).\end{cases}$$
  Moreover, for a cycle cover $\mathcal{C}$, let $\mu'(\mathcal{C}):=\prod_{C\in\mathcal{C}} \mu'(C)$.
  By the Leibniz formula and the fact that only one entry of row $v$ (column $u$) is non-zero,
  we have
  $$\det(\asymb^{v,u}(G))=\sum_{\mathcal{C}\in K(G)} \mu'(\mathcal{C}).$$

  Let $\mathcal{C}$ be some cycle cover of $G$. If $vu\notin E(\bigcup\mathcal{C})$,
  then $\mathcal{C}$ contributes $0$ to the above sum by the definition of $\mu'(\mathcal{C})$.
  On the other hand, if $vu\in E(\bigcup\mathcal{C})$, then there exists some $C\in\mathcal{C}$
  that consists of a simple path $P_C=u\to v$ and an edge $vu$.
  Note that, by the definition of $\mu(P_C)$, and $\mu'(C)=\mu(C)/x_{v,u}$, we have
  $\mu'(C)=\mu(P_C)$. Moreover, $\mathcal{C}\setminus\{C\}$ is a cycle cover
  of $G[V\setminus V(P_C)]$.
  So we have $\mu'(\mathcal{C})=\mu(P_C)\cdot \mu'(\mathcal{C}\setminus\{C\})=\mu(P_C)\cdot \mu(\mathcal{C}\setminus \{C\})$.

  There is a 1-1 correspondence between cycle covers $\mathcal{C}$ with $vu\in E(\bigcup\mathcal{C})$
  and pairs $(P,\mathcal{C}')$, where $P\in\mathcal{P}_{u,v}$ and $\mathcal{C}'\in K(G[V\setminus V(P)])$.
  As a result, we have
  \begin{align*}
    \det(\asymb^{v,u}(G))&=\sum_{\substack{\mathcal{C}\in K(G)\\ vu\in E(\bigcup\mathcal{C})}}\mu'(\mathcal{C})\\
                    &=\sum_{P\in \mathcal{P}_{u,v}}\sum_{\mathcal{C}'\in K(G[V\setminus V(P)])}\mu(P)\cdot \mu(\mathcal{C}')\\
                    &=\sum_{P\in \mathcal{P}_{u,v}}\mu(P)\sum_{\mathcal{C}'\in K(G[V\setminus V(P)])}\mu(\mathcal{C}')\\
                    &=\sum_{P\in \mathcal{P}_{u,v}}\mu(P)\cdot \det(\asymb(G[V\setminus V(P)])).\hspace{2mm}\qedhere
  \end{align*}
\end{proof}

Since $\asymb(G)$ is invertible, we have
  $$\asymb(G)^{-1}_{u,v}=\frac{(-1)^{u+v}}{\det(\asymb(G))}\cdot \det(\asymb^{v,u}(G)),$$
In particular, as noted by Sankowski~\cite{Sankowski04}, since $\det(\asymb(G[Z]))\not\equiv0$ for any $Z\subseteq V$, it follows by \Cref{l:sumpath}
that $\asymb(G)^{-1}_{v,u}$ is a zero polynomial if and only if $\sum_{P\in \mathcal{P}_{u,v}}\mu(P)$ is a zero polynomial, or in other words, if there is no
$u\to v$ path in $G$.

We are now ready to prove \Cref{t:tr-matrix}.
Recall that $F=\{u_1v_1,\ldots,u_kv_k\}$ is
a group of parallel inter-SCC edges, that is,
$u_i\in R$, $v_i\in T$, where $R$ and $T$
are distinct SCCs of $G$. Let $r\in R$
and $t\in T$ be arbitrarily chosen.

Let us first assume that $F$ is non-redundant.
Then,
for all
simple $r\to t$ paths $P$, $P$ is of the form $P_R\cdot f\cdot P_T$, where
$f=u_iv_i\in F$, $P_R$ is a simple (possibly empty) $r\to u_i$ path
entirely contained in $R$, and $P_T$ is a simple $v_i\to t$ path
entirely contained in $T$.
Note that neither $P$ can use two edges from $F$ at once, nor it can enter $R$ after leaving it for the first time, nor it can leave $T$ once that SCC is entered.
So, by \Cref{l:sumpath}, we have
\begin{align*}
  \det(\asymb^{t,r}(G))
  &=\sum_{P\in \mathcal{P}_{r,t}} \mu(P)\cdot \det(\asymb(G[V\setminus V(P)]))\\
  &=\sum_{i=1}^k \sum_{P_R\in \mathcal{P}_{r,u_i}}\sum_{P_T\in \mathcal{P}_{v_i,t}}\mu(P_R\cdot u_iv_i\cdot P_T)\cdot \det(\asymb(G[V\setminus V(P_R)\setminus V(P_T)]))\\
  &=-\sum_{i=1}^k x_{u_i,v_i}\sum_{P_R\in \mathcal{P}_{r,u_i}}\sum_{P_T\in \mathcal{P}_{v_i,t}}\mu(P_R)\cdot \mu(P_T)\cdot \det(\asymb(G[V\setminus V(P_R)\setminus V(P_T)])).
\end{align*}

Let $S_1,\ldots,S_s$ be the SCCs of $G$. 
Since every path $P_R$ ($P_S$) that the sum iterates through is contained in the SCC $R$ ($T$, resp.), and $R\neq T$,
using \Cref{l:det-scc} and \Cref{l:det-scc-split} applied
to the graphs $G$ and $G[V\setminus V(P_R)\setminus V(P_T)]$, we obtain:
\begin{align*}
  \det(\asymb^{t,r}(G))
  &=-\sum_{i=1}^k x_{u_i,v_i}\sum_{P_R\in \mathcal{P}_{r,u_i}}\sum_{P_T\in \mathcal{P}_{v_i,t}}\mu(P_R)\cdot \mu(P_T)\cdot \left(\prod_{S_i\notin \{R,T\}} \det(\asymb(G[S_i]))\right) \cdot\\
  & \hspace{2cm}\det(\asymb(G[R\setminus V(P_R)])) \cdot \det(\asymb(G[T\setminus V(P_T)]))\\
  &=-\det(\asymb(G))\sum_{i=1}^k x_{u_i,v_i}\left(\frac{1}{\det(\asymb(G[R]))}\sum_{P_R\in \mathcal{P}_{r,u_i}}\mu(P_R)\cdot \det(\asymb(G[R\setminus V(P_R)])) \right)\cdot \\
  &\hspace{2cm}\left(\frac{1}{\det(\asymb(G[T]))}\sum_{P_T\in \mathcal{P}_{v_i,t}}\mu(P_T)\cdot \det(\asymb(G[T\setminus V(P_T)])) \right).
\end{align*}
Therefore, by the relationship between the inverse and the adjoint:
\begin{align}\label{eq:u1v1}
  \begin{split}
  \asymb(G)^{-1}_{r,t}
  &=(-1)^{r+t}\sum_{i=1}^k -x_{u_i,v_i}\left(\frac{1}{\det(\asymb(G[R]))}\sum_{P_R\in \mathcal{P}_{r,u_i}}\mu(P_R)\cdot \det(\asymb(G[R\setminus V(P_R)])) \right)\cdot \\
  &\hspace{2cm}\left(\frac{1}{\det(\asymb(G[T]))}\sum_{P_T\in \mathcal{P}_{v_i,t}}\mu(P_T)\cdot \det(\asymb(G[T\setminus V(P_T)])) \right).
  \end{split}
\end{align}
On the other hand, again by \Cref{l:sumpath,l:det-scc}, we have
\begin{align*}
  \asymb(G)^{-1}_{r,u_i}&=\frac{(-1)^{r+u_i}}{\det(\asymb(G))}\cdot \det(\asymb^{u_i,r}(G))\\
  &= \frac{(-1)^{r+u_i}}{\det(\asymb(G))}\cdot\sum_{P\in \mathcal{P}_{r,u_i}} \mu(P)\cdot \det(\asymb(G[V\setminus V(P)]))\\
  &=\frac{(-1)^{r+u_i}}{\det(\asymb(G))}\cdot\sum_{P\in \mathcal{P}_{r,u_i}} \mu(P)\cdot \det(\asymb(G[R\setminus V(P)]))\cdot \prod_{S_i\neq R} \det(\asymb(G[S_i]))\\
  &=\frac{(-1)^{r+u_i}}{\det(\asymb(G))}\cdot \frac{\det(\asymb(G))}{\det(\asymb(G[R]))}\sum_{P\in \mathcal{P}_{r,u_i}} \mu(P)\cdot \det(\asymb(G[R\setminus V(P)])).
\end{align*}
and thus
\begin{equation}\label{eq:u1ui}
  (-1)^{r+u_i}\cdot \asymb(G)^{-1}_{r,u_i}=\frac{1}{\det(\asymb(G[R]))}\sum_{P\in \mathcal{P}_{r,u_i}} \mu(P)\cdot \det(\asymb(G[R\setminus V(P)])).
\end{equation}
Similarly, we can obtain
\begin{equation}\label{eq:viv1}
  (-1)^{t+v_i}\cdot \asymb(G)^{-1}_{v_i,t}=\frac{1}{\det(\asymb(G[T]))}\sum_{P\in \mathcal{P}_{v_i,t}} \mu(P)\cdot \det(\asymb(G[T\setminus V(P)])).
\end{equation}
By plugging in \Cref{eq:u1ui,eq:viv1} into \Cref{eq:u1v1}, we obtain:

\begin{align*}
  \asymb(G)^{-1}_{r,t}&=(-1)^{r+t}\sum_{i=1}^k -x_{u_i,v_i}\cdot (-1)^{r+u_i}\cdot \asymb(G)^{-1}_{r,u_i} \cdot (-1)^{t+v_i}\cdot \asymb(G)^{-1}_{v_i,t}\\
  &=-\sum_{i=1}^k (-1)^{u_i+v_i}\cdot x_{u_i,v_i}\cdot \asymb(G)^{-1}_{r,u_i}\cdot \asymb(G)^{-1}_{v_i,t},
\end{align*}
which proves the ``$\impliedby$'' implication of \Cref{t:tr-matrix}.

Now suppose the edges $F$ are redundant. It is enough to prove
\begin{align*}
  \det(\asymb(G))\cdot \asymb(G)^{-1}_{r,t}\not\equiv-\det(\asymb(G))\sum_{i=1}^k (-1)^{u_i+v_i}\cdot x_{u_i,v_i}\cdot \asymb(G)^{-1}_{r,u_i}\cdot \asymb(G)^{-1}_{v_i,t},
\end{align*}
or, by multiplying both sides by $\det(\asymb(G))$, equivalently
\begin{align}\label{eq:cond}
  \det(\asymb(G))\cdot \det(\asymb^{t,r}(G))\not\equiv-\sum_{i=1}^k x_{u_i,v_i}\cdot \det(\asymb^{u_i,r}(G))\cdot \det(\asymb^{t,v_i}(G)),
\end{align}
We prove that the polynomial on the left-hand side of \Cref{eq:cond}
contains a monomial
that the right-hand side polynomial lacks.
Namely, let $P$ be some simple $r\to t$ path that goes through an SCC $Y$ of $G$
such that $Y\neq R$ and $Y\neq T$. 
Note that $P$ does not go through any of edges in $F$, since all vertices
of $P\cap Y$ have to appear on $P$ after all vertices of $P\cap R$, and before all vertices
of $P\cap T$.
By \Cref{l:sumpath},
the left-hand side contains a monomial
$$\left(\prod_{i=1}^n x_{i,i}\right)\cdot \mu(P) \cdot \left(\prod_{i\in V\setminus V(P)} x_{i,i}\right).$$
However, each monomial in the right-hand side polynomial has a variable
of the form $x_{u_i,v_i}$, where $u_iv_i\in F$, whereas the above
monomial clearly does not contain such variables by $E(P)\cap F=\emptyset$.

}

\appendix

\section{Decremental Single Source Reachability on DAGs} \label{app:dag}

In this section, 
we explain a decremental data structure that maintains single source reachability information 
on \DAG{G = (V, E)}, as summarized in the following lemma.
The data structure extends that of \citeauthor{Italiano:1988aa}~\cite{Italiano:1988aa}
and is equipped with the operations required in \Cref{subsec:combinatorial_dag}.

\italiano*

\paragraph*{The Data Structure.}

For every vertex \( y \neq r \), we define a doubly linked list \activee{y}{r} consisting of incoming edges of \( y \) in \( G \).
The data structure maintains \parent{y}{r}, which points to the first edge in \activee{y}{r} that connects \( y \) to \desc{r}, and \cc{y}{r}, which points to the second edge in \activee{y}{r} that connects \( y \) to \desc{r}.
If no such edge exists in \activee{y}{r},
we set the respective pointer to be \nul.

We now introduce the two invariants of the data structure. 
\begin{invariant} \label{invar:tree}
For every vertex \( y \neq r \), if \( \parent{y}{r} = \nul \), then \( y \) is not reachable from \(r\).
\end{invariant}

\begin{invariant} \label{invar:other}
For every vertex \( y \neq r \),  if \( \cc{y}{r} = \nul \), then \( r \) can reach \(y\) through at most one edge.
\end{invariant}

\paragraph*{Initialization.}
To compute \desc{r}, 
we simply compute a reachability tree rooted at \( r \).
For every vertex \( y \neq r \),
we set \activee{y}{r} to be the list of all incoming edges of \( y \) in \( G \).

After moving a vertex in \(\activee{y}{r} \cap \desc{r}\)
to the front of \activee{y}{r} (if such a vertex exists),
we set \parent{y}{r} to point to the front of \activee{y}{r}.
To initialize \cc{y}{r}, 
we first point \cc{y}{r} to the front element of \activee{y}{r}, 
and then call \updateC{y}{r} in \Cref{ds:app} 
to find the first edge in \activee{y}{r} satisfying the definition of \cc{y}{r}. 

Set \(\D{r}\) would maintain the set of vertices removed from \desc{r} due to the last update.
We initialize \(\D{r} = \emptyset\).

\begin{algorithm}[]
\DontPrintSemicolon
\caption{Decremental-Single-Source-Reachability-on-DAGs}
\label{ds:app}

\KwIn{a \DAG{G = (V, E)} and a root vetex \(r \in V\)}
\Maintain{set \(\desc{r} \subseteq V\) of vertices reachable from \(r\) and set \D{r} of vertices removed from \desc{r} due to the last deletion}

\Procedure{Initialize}{

\( \desc{r} \gets \text{vertices reachable from \( r \)} \)

\ForEach{vertex \( y \neq r \)}{

\(\activee{y}{r} \gets\) incoming edges of \(y\) in \(E\)

move a vertex (if any) in \( \activee{y}{r} \cap \desc{r} \) to the front of \activee{y}{r}

set \parent{y}{r} and \cc{y}{r} to point to the front of \activee{y}{r}

\updateC{y}{r}.

}

}

\Procedure{Delete(\Ed)}{

\( E \gets E \setminus \Ed \)

\( Q \gets \Ed \)

\( \D{r} \gets \emptyset \) 

\While{\( Q \neq \emptyset \)}{

\( \edge{x}{y} \coloneqq \textsc{DeQueue(\( Q \))} \)

\uIf{\( \edge{x}{y} = \parent{y}{r} \)}{

\updateP{y}{r}

}
\uElseIf{\( \edge{x}{y} = \cc{x}{r} \)}
{
\updateC{y}{r}
}
\Else{
remove \edge{x}{y} from \activee{y}{r}
}
}
}

\Procedure{UpdateP(\(y\))}{

\eIf{\( \cc{y}{r} = \nul \)}{
\( \parent{y}{r} \gets \nul \)

add \(y\) to \D{r}

\ForEach{outgoing edge \( \edge{y}{z} \in E \)}{
\textsc{EnQueue(\( Q, \edge{y}{z} \))}
}

}{
remove items from the front of \activee{y}{r} until \( \parent{y}{r} = \cc{y}{r} \)

\updateC{y}{r}
}

}

\Procedure{UpdateC(\(y\))}{

move pointer \cc{y}{r} to the next element in \activee{y}{r}

\If{\( \cc{y}{r} = \edge{z}{y} \) such that \( z \neq r\) and \( \parent{z}{r} = \nul \) }{

\updateC{y}{r}

}
}

\Procedure{In(\(y\))}{

\eIf{\parent{y}{r} or \cc{y}{r} points to an edge different than \edge{r}{y}}{
\Return \texttt{True}
}{
\Return \texttt{False}
}
}

\end{algorithm}

\paragraph*{Handling Edge Deletions.}

Assume that a deletion of edges \Ed has happened
as the last update.
To maintain \desc{r}, we use the queue \( Q \),
containing the edges which their tail needs to be reconnected to \desc{r} after the removal of the edge.
We begin by setting \( Q = \Ed \).

While \( Q \neq \emptyset \), for edge \( \edge{x}{y} \in Q \), the algorithm checks whether the removal of \edge{x}{y} affects \Cref{invar:tree,invar:other}.
The maintenance procedure is as follows.
\begin{enumerate}
\item 
If \( \edge{x}{y} = \parent{y}{r} \), then \( y \) loses its connection to \desc{r}.
In this case, we update \parent{y}{r} as follows.
\begin{itemize}
\item  
If \( \cc{y}{r} \neq \nul \), we remove the elements in the front of \activee{y}{r} until we get \( \parent{y}{r} = \cc{y}{r} \).
By the definition of pointers, this ensures that \Cref{invar:tree} is correctly maintained.
To correctly maintain \cc{y}{r}, we update \cc{y}{r} by calling \updateC{y}{r}, which finds the first edge in \activee{y}{r}, after \parent{y}{r}, that can connect \( y \) to \desc{r}, as desired.

\item 
If \( \cc{y}{r} = \nul \), then \edge{x}{y} was the only edge connecting \( y \) to \desc{r}.
Thus, we set \( \parent{y}{r} \gets \nul \), remove \( y \) from \desc{r} and add it to \D{r}.
Since the children of \( y \) may be connected to \desc{r} through \( y \), we add all outgoing edges of \( y \) to \( Q \).
\end{itemize}

\item
If \( \edge{x}{y} = \cc{y}{r} \), we update \cc{y}{r} by calling \updateC{y}{r} to correctly maintain \cc{y}{r},
and then remove \edge{x}{y} from \activee{y}{r} to correctly maintain \activee{y}{r} as a subset of incoming edges of \( y \).

\item
If \( \edge{x}{y} \neq \parent{y}{r} \) and \( \edge{x}{y} \neq \cc{y}{r} \), we only need to remove \edge{x}{y} from \activee{y}{r}.
\end{enumerate}

We conclude this section by the following lemma.

\begin{lemma} \label{lem:app}
Given an \(m\)-edge \DAG{G = (V, E)}, there is a decremental data structure that maintains the set \desc{r} of vertices reachable from the root vertex $r$ in \( O(m) \) total update time, and supports the following additional operation:
	\begin{itemize}
    	\item 
    	\textsc{In\( (y) \):} Return \texttt{True} if \( y \neq r \) and \( y \) has an in-neighbor from \( \desc{r} \setminus r \), and \texttt{False} otherwise. 	\end{itemize}
	Additionally, it maintains the set \D{r} of vertices that have been removed from \desc{r} due to the most recent deletion.
\end{lemma}

\begin{proof}

\underline{Correctness:}
follows from \cite{Italiano:1988aa} and observing that the data structure correctly maintains the pointers \parent{\cdot}{r} and \cc{\cdot}{r} as discussed above.

\underline{Update time:}
the update time is dominated by the time needed to maintain \parent{\cdot}{r} and \cc{\cdot}{r}.
By \cite{Italiano:1988aa}, the total time required to maintain \parent{\cdot}{r} over any sequence of edge deletions is \( O(m) \).
Note that, for every \( y \neq r \), \cc{y}{r} and \parent{y}{r} probe \activee{y}{r} at most once in total.
Since \activee{y}{r} consists of the incoming edges to \( y \), we conclude that the total update time to maintain \parent{\cdot}{r} and \cc{\cdot}{r} is bounded by \( O\left( \sum _{u \in V} \deg{v} \right) = O(m) \).
\end{proof}

\paragraph*{Extending the Data Structure.}
We extend the data structure of \Cref{lem:app} to support all the operations of \Cref{lem:extended_italiano}.
We define \anc{r} to be the set of vertices that can reach \( r \), and \A{r} as the set of vertices that are removed from \anc{r} due to the last deletion of edges \( \Ed \).

\begin{proof}[Proof of \Cref{lem:extended_italiano}]
We use the data structure of \Cref{lem:app} to maintain \desc{r}, return \D{r}, and answer \textsc{In\( (y) \)} for every vertex \( y \neq r \).
To maintain \anc{r}, return \A{r}, and answer \textsc{Out\( (y) \)}, we use the data structure of \Cref{lem:app} on the reverse graph \( G' = (V, E') \) of \( G \) defined by \( \edge{x}{y} \in E \) iff \( \edge{y}{x} \in E' \).
It is easy to see that \anc{r} is equal to the set of vertices reachable from \( r \) in \( G' \).
Therefore, the correctness and the time complexity for each operation immediately follows.
\end{proof}

\section{Decremental Single Source Reachability on General Graphs} \label{app:general}

In this section, we explain a data structure that
maintains a decremental single source reachability tree \T{r} rooted at \( r \) on a graph \( G = (V, E) \).
The guarantees are stated in the following lemma.

\roditty*

Our algorithm builds upon that of \citeauthor{Roditty:2016aa}, which we briefly explain in \Cref{subsec:ds_RZ} 
before proceeding with our extension in \Cref{subsec:ds_app_general}.

\subsection{The data structure of \cite{Roditty:2016aa}} \label{subsec:ds_RZ}

The data structure uses \textit{uninspected} inter-SCC edges to maintain \T{r}.
An uninspected edge is either a tree edge, which will remain uninspected, or it is not useful; once it has been inspected by the data structure, it will no longer be uninspected.
Initially, all edges are uninspected.

A vertex \( z \) is called \textit{active} if \( \activee{Z}{r} \neq \nul \).
For a vertex \( z \in V \), \inn{z}{r} is the set of uninspected inter-SCC edges entering \( z \), 
and \outt{z}{} is the set of \textit{all} edges outgoing \( z \).
For every SCC \( Z \), 
the data structure maintains a doubly linked list \activee{Z}{r} that contains all active vertices of \( Z \).

The data structure maintains the following invariant.
\begin{invariant} \label{invar:tree_general}
If \( y \in V \setminus R \) is the first vertex in \activee{Y}{r} and \edge{x}{y} is the first edge in \inn{y}{r}, 
then \edge{x}{y} is the tree edge connecting the SCC \( Y \) to \T{r}.
In particular, if \( \activee{Y}{r} = \nul \), then none of the vertices of \( Y \) are connected to \T{r}.
\end{invariant}

The data structure maintains a sequence of graphs \( G_0, G_1, \dots, G_t \), where \( t \) is the number of insert operations.
Here, \( G_i \) is the snapshot of \( G \) after the \( i \)th insertion.
However, edges that are subsequently deleted from \( G \) are also deleted from $G_i$, ensuring that at each step,   \( \emptyset = E_0 \subseteq E_1 \subseteq \dots \subseteq E_t = E \), where \( E_i \) is the set of edges in \( G_i \).
The data structure maintains an array indexed by \( V \) to maintain the SCCs: for every vertex \( z \in V \), \scc{z}{i} is the name of the SCC containing \( z \) in \( G_i \).
The guarantees of the data structure are summarized in the following theorem.

\begin{theorem}[Sections 3 and 4 of \cite{Roditty:2016aa}, rephrased] \label{th:roditty_SCC}
Given a directed graph \( G=(V, E) \) and the sequence  of subgraphs \( G_0, G_1, \dots, G_t \) defined above, there is a fully dynamic data structure that maintains the SCCs for each \( G_i \), and supports each insert operation on \( G \) in $O(m + n \log n)$ worst-case time and each delete operation on \( G \) in $O(m + n \log n)$ amortized update time,
where \( m \) is the number of edges in the current graph \( G \).

Moreover, the data structure supports the following additional operations:
\begin{itemize}
\item 
\textsc{Detect\( () \):} List all the components that decomposed as a consequence of the most recent delete operation in \( G \), together with the index \( i \) specifying the subgraph \( G_i \) that the decomposition happened.
This operation runs in \( O(n) \) time.

\item 
\textsc{List\( (Z, i) \):} Given an SCC \( Z \) in \( G_i \), list all the SCCs \( Z \) decomposed as a consequence of the most recent delete operation.
This operation runs in time proportional to the number of the SCCs that \( Z \) decomposed into.

\end{itemize}
\end{theorem}

\paragraph*{Handling Edge Insertions.}
After an insertion, we can afford to recompute the SCCs in the current graph, i.e., in \( G_t \), as well as all the values required by the data structure, such as \inn{\cdot}{r} and \outt{\cdot}{}.
It is easy to see that this takes \( O(m + n \log n ) \) time.

\paragraph*{Handling Edge Deletions.}
We explain how to efficiently maintain the SCCs of \( G_i \) after the deletion of \Ed.
Using the operation \textsc{Detect\( () \)} from \Cref{th:roditty_SCC}, we obtain the SCCs of \( G_i \) that has decomposed after the deletion.
Assume that SCC \( Z \) is one of them, and \textsc{List\( (Z, i) \)} has returned \( Z_1, Z_2, \dots, Z_k \) as the SCCs \( Z \) decomposed into, where \( |Z_1| \geq |Z_2| \geq \dots \geq |Z_k| \).
We first show how to maintain \activee{z_j}{r}.
Instead of scanning \activee{Z}{r} and moving each vertex to a new list, we let \( Z_1 \) to inherit \activee{Z}{r}.
Thus, we only need to move the vertices that does not belong to \( Z_1 \), resulting in \( O(\sum _{j = 2} ^ k |Z_j| ) \) time for this operation.
Note that we can simultaneously update \inn{z}{r} for every vertex \( z \in Z \).
Since \( |Z_1| \geq |Z_j| \), during a sequence of edge deletions, each vertex can be moved at most \( \log n \) times, which results in the following lemma.

\begin{lemma}[Lemma 5.1 of \cite{Roditty:2016aa}] \label{lem:active_RZ}
    The total cost to maintain \activee{\cdot}{r} in each graph \( G_i \) during any sequence of edge deletions is \( O(m + n \log n) \), where \( m \) is the number of edges in the initial graph \( G_i \).
\end{lemma}

To maintain \inn{\cdot}{r},
we add each SCC that may need to be inspected to \( Q \).
If SCC \( Z \) has decomposed into \( Z_1, Z_2, \dots, Z_k \), we then add all \( Z_j \) to \( Q \), except for the (possibly) one that contains the root \( r \).
Also, for every edge \( \edge{x}{y} \in \Ed \), if \( y \) is the first vertex in \activee{Y}{r} and \edge{x}{y} is the first edge in \inn{y}{r}, we then add the SCC \( Y \) to \( Q \).
We then pick an SCC \( W \in Q \).
Let \( w \) be the first vertex in \activee{W}{r}.
We scan \inn{w}{r} to find an inter-SCC incoming edge \edge{x}{w} that satisfies the following conditions: it should be removed from \( G_i \), and either \( \activee{\scc{x}{i}}{r} \neq \emptyset \) or \( \scc{x}{i} = \scc{r}{i} \).
We keep removing edges from \inn{w}{r} until we find such an edge.
If we reach \( \inn{w}{r} = \nul \), we remove \( w \) from \activee{W}{r}, choose the first vertex in \activee{W}{r} again, and try to find an edge that connects \( W \) to \T{r}. 
If we reach \( \activee{W}{r} = \nul \) without finding a connecting edge, then \( W \) is no longer connected to \T{r}.
In this case, for every vertex \( w \in W \) and every outgoing edge \( \edge{w}{y} \in \outt{w}{} \), if \( y \) is the first vertex in \activee{Y}{r} and \edge{w}{y} is the first edge in \inn{y}{r}, we then add the SCC \( Y \) to \( Q \).
We conclude this subsection with the following lemma which guarantees the correctness of the algorithm.

\begin{lemma}[Lemma 5.2 of \cite{Roditty:2016aa}] \label{lem:T_RZ}
    The algorithm described above correctly maintains \T{r} in \( G_i \) (and so \inn{\cdot}{r} and \activee{\cdot}{r}) during any sequence of edge deletions in \( O(m) \) total update time, where \( m \) is the number of edges in the initial graph \( G_i \).
\end{lemma}

\subsection{Our extension to the data structure} \label{subsec:ds_app_general}

Let \graph{u} be the snapshot of \( G \) taken after the last insertion \Ei{u} centered around \( u \) (if such an insertion occurred).
Note that subsequent edge insertions centered around vertices \emph{different} from $u$ do not alter $\graph{u}$. 
However, edges that are subsequently deleted from \( G \) are also deleted from \graph{u}, ensuring that at each step, \( \E{u} \subseteq E \).
Thus, \graph{u} undergoes only edge deletions. 
If vertex $u$ serves as an insertion center in future updates, the snapshot graph $\graph{u}$ is reinitialized, and \T{u} is constructed from scratch. 

Note that \graph{u} is actually the graph \( G_i \) in the sequence defined in \Cref{subsec:ds_RZ}, where \( i \) is the \textit{last} time that an insertion happened around \( u \).
We can easily maintain \graph{u}'s by maintaining the sequence \( G_0, G_1, \dots, G_t \) and bookkeeping the time index \( i \) in the sequence for each vertex \( u \).
After an insertion around \( u \), we  update the index to the most recent \( t \).

Similar to \Cref{subsec:combinatorial_general}, we define \desc{u} as the set of vertices reachable from \( u \), and \anc{u} as the set of vertices that can reach \( u \) in \graph{u}.
We extend the data structure of \Cref{subsec:ds_RZ} to support the following additional operations.
\begin{itemize}
    
    \item 
    \textsc{In\( (y, r) \):}
    for any vertex \( y \) with \( Y \neq R \), return \texttt{True} if \( Y \) has an in-neighbor from  \( \desc{r} \setminus R \) In \graph{r}, and \texttt{False} otherwise. 
    Here, \( R, Y \) are the SCCs that contain \( r, y \) in \graph{r}, respectively.
    
    \item
    \textsc{Out\( (x, r) \):}
    for any vertex \( x \) with \( X \neq R \), return \texttt{True} if \( R \) has an out-neighbor to  \( \anc{r} \setminus X \) In \graph{r}, and \texttt{False} otherwise. 
    Here, \( R, X \) are the SCCs that contain \( r, x \) in \graph{r}, respectively.
    
\end{itemize}

Here, we maintain \T{u} in each \graph{u}.
To adapt this to the data structure of \Cref{subsec:ds_RZ}, for each \( G_i \), if it is a snapshot of an insertion around vertex \( u \), we then simply initialize \T{u}, and maintain the related values decrementally.
Note that, by \Cref{lem:active_RZ,lem:T_RZ}, we can freely choose the root vertex \( r \), which in this case is \( r = u \).

To implement the operations, we take advantage of the same idea we used in \Cref{app:dag}, but now on the SCCs instead of the vertices: for each SCC \( Y \), we maintain the two pointers \parent{Y}{r},
which points to an edge \edge{x_2}{y} such that \( x_2 \in V (\T{r}) \) and \( \scc{x_2}{r} \neq \scc{x_1}{r} \),
and \cc{Y}{r}, 
which points to an edge \edge{x_2}{y} such that \( x_2 \in V (\T{r}) \) and \( \scc{x_2}{r} \neq \scc{x_1}{r} \).
If no such edge exists in \parent{Y}{r} or \cc{Y}{r},
we set the respective pointer to be \nul.
Roughly speaking, \parent{Y}{r} and \cc{Y}{r} maintain two \emph{different} (if any) incoming edges  \edge{x_1}{y} and \edge{x_2}{y} of \( Y \) such that \( x_1 \) and \( x_2 \) belong to \( V(\T{r}) \) from different SCCs.

The data structure maintains \activee{Y}{r} such that \parent{Y}{r} always points to the front of \inn{y}{r}, where \( y \) is the first vertex of \activee{Y}{r}.
Also, \cc{Y}{r} always points to the first edge \textit{after} \inn{y}{r} that satisfies the criteria defining \cc{Y}{r}.
We now introduce the other invariant of our data structure.

\begin{invariant} \label{invar:other_general}
If \( y \in V \setminus R \), then  \cc{Y}{r} is an edge coming from an SCC different than $\parent{Y}{r}$ and connects $Y$ to $\T{r}$ (if any).
If \( \cc{Y}{r} = \nul \), then \( Y \) has at most one SCC connecting it to \T{r}.
\end{invariant}

Similar to \Cref{app:dag}, handling the sets \anc{r} and  \A{r}, and the operation \textsc{Out\( (x, r) \)} is simply done on the reverse graph of \graph{u}.
Thus, here, we only discuss how to handle the sets \desc{r} and  \D{r}, and the operation \textsc{In\( (y, r) \)} as follows.
\begin{itemize}
    \item 
    Handling \( \desc{u} = V(\T{u}) \) is simply done by the data structure.
    If an SCC \( Y \) got removed during the update, i.e., we reach \( \activee{Y}{r} = \nul \), we then add all the vertices of \( Y \) to \D{u}.
    \item 
    The operation \textsc{InNeighbor\( (y, r) \)} is supported by the pointer \cc{Y}{r}.
    By the definition of the pointer, if \( \cc{Y}{r} = \nul \), we return \texttt{False} in response to the operation, and \texttt{True} otherwise.
\end{itemize}

Handling the insertions and deletions are similar to the data structure of \Cref{app:dag}.

\clearpage

\section*{References}
\printbibliography[heading=none]

\clearpage

\end{document}